%% file: main_CDC.tex
\documentclass[letterpaper, 10 pt, conference]{ieeeconf}  

\IEEEoverridecommandlockouts 
\usepackage{graphicx}
\usepackage{subfigure}
 \usepackage[top=2cm, bottom=2cm, left=2.3cm, right=2.3cm]{geometry}
 \usepackage{algorithm}
  \usepackage{algorithmic}
 \usepackage{amsfonts}
 \usepackage{amsmath}
 \usepackage{amsthm}
 \usepackage{amssymb}
 \usepackage{float}
 \usepackage{caption}
 \usepackage{multirow}
 \usepackage{wrapfig}
 \usepackage{hyperref}
 \DeclareMathOperator*{\argmax}{argmax}
 
 \usepackage{cite} 
 \usepackage{xcolor}

\newcommand{\NE}{\textup{NE}}
\newcommand{\local}{\textup{local}}
\newcommand{\cM}{\mathcal{M}}
\newcommand{\cS}{\mathcal{S}}
\newcommand{\cA}{\mathcal{A}}
\newcommand{\PoA}{\textup{PoA}}
\newcommand{\bE}{\mathbb{E}}

\newcommand{\oQ}{\overline{Q}}
\newcommand{\hQ}{\widehat{Q}}
\newcommand{\oQh}{\overline{Q}_{i,h}^\pi}
\newcommand{\hQh}{\widehat{Q}_{i,h}^\pi}
\newcommand{\oA}{\overline{A}}
\newcommand{\orh}{\overline{r}_{i,h}^\pi}
\newcommand{\hrh}{\widehat{r}_{i,h}^\pi}
\newcommand{\hPh}{\widehat{P}_{i,h}}
\newcommand{\Ph}{P_{i,h}}
\newcommand{\NEgap}{\textup{NE-Gap}}
\newcommand{\placeholder}{4n^2H^3}

\newtheorem{theorem}{Theorem}
\newtheorem{lemma}{Lemma}
\newtheorem{rmk}{Remark}
\newtheorem{defi}{Definition}
\newtheorem{assump}{Assumption}

\title{\LARGE \bf
Markov Games with Decoupled Dynamics: Price of Anarchy and Sample Complexity}

\author{Runyu (Cathy) Zhang, Yuyang Zhang, Rohit Konda, Bryce Ferguson, Jason Marden, Na Li
\thanks{Runyu (Cathy) Zhang, Yuyang Zhang and Na Li are with the John A. Paulson School of Engineering and Applied Sciences, Harvard University (e-mail: runyuzhang@fas.harvard.edu, yuyangzhang@g.harvard.edu, nali@seas.harvard.edu). Rohit Konda, Bryce Ferguson, Jason Marden are with the Department of Electrical and Computer Engineering, University of California, Santa Barbara(e-mail: rkonda@ucsb.edu, blferguson@ece.ucsb.edu, jrmarden@ece.ucsb.edu).}
}

\allowdisplaybreaks
\begin{document}

\maketitle
\thispagestyle{empty}
\pagestyle{empty}

\begin{abstract}

This paper studies the finite-time horizon Markov games where the agents' dynamics are decoupled but the rewards can possibly be coupled across agents. The policy class is restricted to local policies where agents make decisions using their local state. We first introduce the notion of smooth Markov games which extends the smoothness argument for normal form games (\cite{roughgarden2015intrinsic,Chandan19}) to our setting, and leverage the smoothness property to bound the price of anarchy of the Markov game. For a specific type of Markov game called the Markov potential game, we also develop a distributed learning algorithm, multi-agent soft policy iteration (MA-SPI), which provably converges to a Nash equilibrium. Sample complexity of the algorithm is also provided. Lastly, our results are validated using a dynamic covering game. 
\end{abstract}

\input{CDC2023/sections/1_intro}

\input{CDC2023/sections/2_prelim}

\input{CDC2023/sections/3_smoothness_poa}

\input{CDC2023/sections/4_alg_complexity}

\input{CDC2023/sections/5_proof_sketch}

\input{CDC2023/sections/6_example}

\input{CDC2023/sections/7_conclusion}




\bibliographystyle{ieeetr}
\bibliography{CDC2023/bib.bib}
\newpage

\onecolumn
\input{CDC2023/sections/8_appendix}

\end{document}

%% file: CDC2023/sections/1_intro.tex
\section{INTRODUCTION}

Multiagent Markov Decision Processes (MDPs) have found numerous applications, such as autonomous vehicles~\cite{isele2018navigating}, swarm robotics~\cite{huttenrauch2019deep}, collaborative manufacturing~\cite{xia2021digital}, decentralized energy management~\cite{kofinas2018fuzzy}, and social network analysis~\cite{zhang2017social}, among many others. In these problems, a critical question is how individual agents can learn effective strategies collectively in complex multi-agent environments.

Given the success of reinforcement learning (RL) for MDPs, many studies have focused on the empirical and theoretical performance of RL in multi-agent settings (e.g., \cite{resQ, wang2021shaq, SongWhencanwe, JinV-learning, ZhangMPG, Leonardos, li2022congestion, calderone2017infinite}). The study of multi-agent  RL is inherently more complex than that of single-agent RL due to the (strategic) interactions between agents. Consequently, most provable algorithms in multi-agent RL (e.g., \cite{SongWhencanwe, JinV-learning, ZhangMPG, Leonardos}) only consider convergence to Nash equilibria (NEs) or their extensions (e.g., coarse correlated equilibria (CCE)) instead of global optimality. However, in many cases, NE may have poor performance because it only considers the individual rationality of each agent and hence does not necessarily result in the best overall outcome for the system as a whole. Therefore, in addition to understanding the convergence to a NE, it is also crucial to investigate the quality of the NEs compared to the global optimal solution, where all agents in the system coordinate and cooperate to achieve a common objective. 

Due to the difficulty in characterizing which specific NE (or CCE) an algorithm converges to, one way to characterize the quality of solutions is to study the worst-case performance of NEs, such as Price of Anarchy (PoA, \cite{koutsoupias1999worst}), which is defined as the ratio between the performance of the worst NE and the global optimum. The PoA has been well studied in the \textit{static game} settings where there is only a single stage/state. In many applications such as traffic routing~\cite{youn2008price}, resource allocation \cite{paccagnan2019utility, gairing2009covering, gkatzelis2016optimal}, and auctions \cite{hartline2014price}, PoA bounds have been established.  Moreover, \cite{roughgarden2015intrinsic} introduces a general approach for generating the PoA bounds for a class of games called smooth games. 
However, it is unclear how to extend these results to the Markov settings where there are many states and stages.  

Several recent works, such as \cite{radanovic2019learning} and \cite{chen2022convergence}, have extended the smoothness property to the Markov game settings. However, one limitation of their approach is that they define smoothness directly on the value functions, which are long-term accumulative rewards, of the Markov game. Value functions are typically hard to compute, making it difficult to verify the smoothness condition. These studies suggest that analyzing the PoA for Markov games may be fundamentally more challenging than for static games due to the increased complexity of the problem structure. 

\textbf{Our Contributions: } Given the difficulty of studying PoA for general Markov games, this paper focuses on a specific type of Markov games, where one agent's state transition only depends on its own state and action but everyone's stage reward is coupled with each other's state and action. This setting finds many applications such as multi-robots, traffic systems, and sensor networks. 

The first part of this paper is dedicated to addressing the challenges associated with the study of PoA, where we provide both positive and negative results. In the case of the former, we extend the smoothness arguments presented in previous work \cite{roughgarden2015intrinsic, Chandan19} for normal form static games, to identify a sufficient condition for a lower bound on the PoA. This is accomplished by taking advantage of the unique structure of our problem, which allows for verification of smoothness solely for the stage rewards, as opposed to the value functions. Compared with existing conditions presented in \cite{chen2022convergence,radanovic2019learning}, our smoothness condition is significantly easier to verify. In terms of negative results, we provide a counterexample to illustrate that the PoA of the stage rewards alone fails to lower-bound the PoA of the Markov game, thus emphasizing the importance of our proposed smoothness condition.

Furthermore, we develop a novel distributed learning algorithm - multi-agent soft policy iteration (MA-SPI) - for a specific type of Markov game called the Markov potential game (MPG), which provably converges to a Nash equilibrium. Notably, MA-SPI can be implemented in a fully decentralized manner, where agents only require local information to update their policies. We also proved that it takes $\widetilde O\left(\frac{n^4H^7}{c^6\epsilon^4} + \frac{n^2\max_i\!|S_i|^2\max_i\!|A_i|^2H^4}{c^2\epsilon^2}\right)$ samples to learn a policy that is $\epsilon$ close to a NE. Here $n$ is the number of agents, $H$ is the horizon length, $|S_i|,|A_i|$ is the size of $i$-th agent's local state and action space and the constant $c$ represents the level of sufficient exploration of the system states (rigorously defined in Assumption \ref{assump:sufficient-exploration}).

In addition to theoretical analysis, we also evaluate our results using a dynamic covering game as an example, where a group of collaborative agents work together to cover as many rewards as possible. The game's smoothness and price of anarchy were analyzed for three different reward designs: identical-interest, marginal-contribution, and utility-sharing. Additionally, we apply the MA-SPI algorithm to the dynamic covering game and numerically show that MA-SPI converges to a Nash equilibrium.

%% file: CDC2023/sections/2_prelim.tex
\section{PROBLEM SETTINGS AND PRELIMINARIES}
\subsection{Markov games with decoupled dynamics}
In this paper, we consider a Markov game model with decoupled dynamics 
\begin{equation}\label{eq:MG-def}
    \cM:= \{\{\cS_i,\cA_i,P_i,r_i,\rho_i\}_{i=1}^n,v,H\},
\end{equation}
 where there are $n$ agents, and $H$ is the horizon of the Markov game.  $s_{i,h}\in\cS_i$ is the local state of agent $i$ at horizon $h$ and $a_{i,h}\in\cA_i$ is the local action. The dynamics of agents are decoupled in the sense that agent $i$'s next state $s_{i,h+1}$ is fully determined by the local action and the current local state $s_{i,h+1}\sim P_{i,h}(\cdot|s_{i,h}, a_{i,h})$. 
 We denote the global state and global action as the concatenation of local states and actions, i.e., $s_h = \{s_{1,h}\dots,s_{n,h}\}, a_h = \{a_{1,h}\dots,a_{n,h}\}$, and thus the state/action space is a product space of the local state/action spaces $\cS = \cS_1\times\cdots\times\cS_n, \cA = \cA_1\times\cdots\times\cA_n$. Agent $i$'s stage reward at horizon $h$ is $r_{i,h}:\cS\times\cA\to[0,1]$, which can possibly depend on other agents' states and actions. Apart from individual rewards, there's a social welfare function $v_h:\cS\times\cA\to\mathbb{R}$ that measures the social welfare of the system at horizon $h$. $\rho_i$ is the initial local state distribution 
 for agent $i$.  A stochastic policy $\pi = \{\pi_h: \cS\rightarrow\Delta(\cA)\}_{h=1}^H$ (where $\Delta(\cA)$ is the probability simplex over $\cA$) specifies a strategy in which agents choose their actions based on the current state in a stochastic fashion.
 Throughout the paper, we consider the \textit{local policy class} $\Pi^\local$, where agents need to take actions independently based on their own local state, i.e., $a_{i,h}\sim \pi_{i,h}(\cdot|s_{i,h})$.
Within the local policy class, the \textit{joint }policy can be written as the product of individual policies, i.e. $\pi_h(a_h|s_h) = \prod_{i=1}^n\pi_{i,h}(a_{i,h}|s_{i,h})$. For a given policy agent $i$'s total reward is denoted as
\begin{equation*}J_i(\pi):=\bE_{s_{j,1}\sim\rho_j}^\pi \sum_{h=1}^{H} r_{i,h}(s_h,a_h),
\end{equation*}
and the total social welfare is
\begin{equation}W(\pi):=\bE_{s_{i,1}\sim\rho_i}^\pi \sum_{h=1}^{H} v_{h}(s_h,a_h).
\end{equation}
Here the notation $\bE^\pi$ denotes the expectation taken over the trajectory by implementing policy $\pi$. Agent $i$'s objective is to maximize its own total reward $J_i$. 

\begin{rmk}[Justification and Limitation of the setting]
The settings described above can find many applications such as sensor coverage, autonomous vehicles, and multi-robotics. In these applications, individual agents (i.e., sensors, vehicles, and robots) have their own dynamics which depend on their own state and action. In Section~\ref{sec:covering}, we provide details on modeling a dynamic covering game using our framework. Nevertheless, one major limitation in our current setting is that the policy class is constrained to fully localized policy whereas in practice, agents could decide their actions based on more information such as neighboring agents' states. We leave it as our important future work to generalize the policy class. \hfill $\square$
\end{rmk}

Throughout the paper, we use the index notation $-i$ to denote all the agents other than agent $i$, for example $\pi_{-i}$ is used to denote $(\pi_1, \dots, \pi_{i-1}, \pi_{i+1}, \dots, \pi_n)$.

 \begin{defi}[Nash equilibrium (NE)]
A policy $\pi^\NE = \{\pi_1^\NE, \cdots, \pi_n^\NE\}$ is called a Nash equilibrium (NE) if for any $i\in [n]$,
\begin{equation*}
    J_i(\pi^\NE) \ge J_i(\pi_i, \pi_{-i}^\NE)
\end{equation*}
We also denote the set of Nash equilibria as $\Pi^{\NE}$.

Given a policy $\pi$, we denote the Nash gap for the policy by:
\begin{align*}
    \NEgap_i(\pi)&:= \max_{\pi_i'} J_i(\pi_i', \pi_{-i}) - J_i(\pi_i, \pi_{-i}).\\
    \NEgap(\pi)&:= \sum_{i=1}^n \NEgap_i(\pi)
\end{align*}
A policy $\pi$ is called an $\epsilon$-NE if $\NEgap(\pi)\le\epsilon$.
\end{defi}

\begin{defi}[Price of anarchy (PoA)] For a Markov game $\cM$ defined as Equation \eqref{eq:MG-def},
\begin{equation*}
   \PoA(\cM):= \frac{\min_{\pi^\NE\in\Pi^\NE}W(\pi^\NE)}{\max_{\pi\in\Pi^\local}W(\pi)}
\end{equation*}
\end{defi}
From the definition of PoA, we can see that for any $\pi^{\NE}$, the performance can be lower bounded by PoA, that is,
$$W(\pi^\NE) \geq \PoA(\cM) \cdot \max_{\pi\in\Pi^\local}W(\pi).$$

In this paper, we mainly focus on two goals: i) firstly, to provide a lower bound for PoA for a class of Markov games so that the bound could be used to bound the performance ratio for any $\pi^\NE$; ii) secondly, to design a distributed reinforcement learning algorithm that enables agents to reach a NE by only using their local information (local states $s_i$, actions $a_i$ and local rewards $r_i$), even without knowing the Markov model $\mathcal{M}$. 

\subsection{Other preliminaries}
Throughout this paper, we will need the following notations and definitions.

We define the state distribution of agent $i$ as $d^\pi_{i,h}:\cS_i \to [0,1]$ such that $d^{\pi_i}_{i,h}(s_i)$ denotes the probability of state $s_i$ being visited at horizon $h$ by implementing policy $\pi$, i.e. $d^{\pi_i}_{i,h}(s_i) = \Pr\{s_{i,h} = s_i|s_{i,1}\sim\rho_i, a_{i,\tau}\sim\pi_i(\cdot|s_{\tau})\}$.
Because the transition probabilities are decoupled and local policies are considered, we have that the total state distribution can be written as:
$d^\pi_h(s) = \Pr\{s_{h} = s\}= \prod_i d^{\pi_i}_{i,h}(s_i).$

Agent $i$'s $Q$-functions is defined as
\begin{equation*}
Q_{i,h}^\pi(s,a):=\bE^\pi \left[\sum_{h'=h}^{H} r_{i,h'}(s_{h'},a_{h'})~\Big|s_h = s, a_h = a\right]
\end{equation*}
Similar to \cite{ZhangMPG}, we define the following `averaged' Q-functions apart from the standard definition of $Q$-functions, which will play an important role in later algorithm design. 

\begin{equation}
   \oQ^\pi_{i,h}(s_i,a_i)
    :=\hspace{-8pt}\mathop{\bE}_{\substack{j\neq i\\ s_j\sim d_{j,\!h}^{\pi_j}(\cdot)}} \mathop{\bE}_{\substack{j\neq i\\ a_j\sim {\pi_j}(\cdot|s_j)}} \hspace{-8pt} Q_{i,h}^\pi(s_i,s_{-i}, a_i, a_{-i})\label{eq:averaged-Q}
\end{equation}
Note that our definition of the ``averaged" Q-function differs from that in \cite{ZhangMPG}. Specifically, we not only take a weighted average over other agents' actions but also their local states.

%% file: CDC2023/sections/3_smoothness_poa.tex
\section{SMOOTH MARKOV GAMES AND PRICE OF ANARCHY}
Characterizing the exact PoA can be challenging even in the normal form game setting, as it compares the performance ratio of the \emph{worst} NE against the global optimal. The smoothness argument \cite{roughgarden2015intrinsic} serves as a useful tool for PoA analysis which can give a canonical PoA lower bound. In this section, we first identify a class of Markov games named {$(\lambda,\mu)$-generalized smooth Markov games}, which is an extension of the smoothness argument in \cite{Chandan19,roughgarden2015intrinsic} for normal form games, and show that {$(\lambda,\mu)$-generalized smooth Markov games} have a uniform lower bounds on the PoA. 

\begin{defi}[{($\lambda,\mu$)-generalized smoothness}]\label{defi:smoothness} For $\mu \ge 0, \lambda > 0$, a Markov game $\cM$ defined as in Equation \eqref{eq:MG-def} is called {$(\lambda,\mu)$-generalized smooth} if the following inequality holds for any state and action pairs $(s_h,a_h), (s_h^\star, a_h^\star)$
    \begin{equation*}
    \begin{split}
        &\sum_{i=1}^n r_{i,h}(s_h,a_h) - r_{i,h}(s_{i,h}^\star, s_{-i,h}, a_{i,h}^\star, a_{-i,h}) \\
        &\le (1+\mu)v_h(s_h,a_h) - \lambda v_h(s_h^\star, a_h^\star),~~\forall h\in [H],
    \end{split}
    \end{equation*}
\end{defi}

\textcolor{black}{Definition \ref{defi:smoothness} is an extension of the generalized smoothness introduced in \cite{Chandan19} to the Markov game setting. The definition implies that at each horizon $h$, the normal form game formed induced by stage reward $\{r_{i,h}\}_{i=1}^n$ and social welfare function $v_h$ satisfies the generalized smoothness argument in \cite{Chandan19}. Notably, this condition is much easier to verify than the smoothness conditions in \cite{radanovic2019learning, chen2022convergence} as it only requires checking the reward and welfare functions at each horizon without any value function computation. The following theorem shows that we can obtain a similar bound on PoA for the extension to the Markov game setting.}

\begin{theorem}\label{thm:PoA-smooth-games}
    Any {$(\lambda,\mu)$-generalized smooth Markov game} $\cM$ satisfies the following PoA bound:
    \begin{equation*}
        \PoA(\cM) \ge \frac{\lambda}{(1+\mu)}
    \end{equation*}
\begin{proof}
We use $\pi^\star = (\pi_1^\star, \dots,\pi_2^\star)$ to denote the globally optimal policy in the local policy class. From the definition of {$(\lambda,\mu)$-generalized smoothness},
    \begin{align}
    &0 \le \sum_{i=1}^n (J_i(\pi^\NE) - J_i(\pi_i^\star,\pi_{-i}^\NE))\label{eq:proof-1}\\
    &=\sum_{i=1}^n\bE_{\pi^\NE} \sum_{h=1}^H r_{i,h}(s_h,a_h) - \notag\\
    &\quad \bE_{(\pi_i^\star,\pi_{-i}^\NE)}\sum_{h=1}^H r_{i,h}(s_{i,h}^\star,s_{-i,h},a_{i,h}^\star,a_{-i,h})\label{eq:proof-2}\\
    &=\bE_{s\sim\pi^\NE, s^\star\sim\pi^\star} \sum_{h=1}^H\Bigg(\sum_{i=1}^n(r_i(s_h,a_h) \notag\\
    &\quad - r_i(s_{i,h}^\star,s_{-i,h},a_{i,h}^\star,a_{-i,h}))\Bigg)\label{eq:proof-3}\\
    &\le{\bE}_{s\sim\pi^\NE, s^\star\sim\pi^\star} \!\sum_{h=1}^H\!\left((\mu \!+\! 1)v(s_h,a_h) \!-\! \lambda v(s_h^\star,a_h^\star)\right)\label{eq:proof-4}\\
    &= (\mu+1) W(\pi^\NE) - \lambda W(\pi^\star),\notag
\end{align}
where the inequality \eqref{eq:proof-1} is derived by the definition of NE, \eqref{eq:proof-2} to \eqref{eq:proof-3} holds because the trajectory $\{s_{i,h}, a_{i,h}\}_{h=1}^H$ of agent $i$ is independent of any other agents $\{s_{-i,h}, a_{-i,h}\}_{h=1}^H$, \eqref{eq:proof-3} to \eqref{eq:proof-4} holds because of the smoothness condition. Thus we have
\begin{equation*}
    \frac{W(\pi^\NE)}{W(\pi^\star)}\ge \frac{\lambda}{\mu + 1}, ~~\Longrightarrow ~\PoA(\cM) \ge \frac{\lambda}{\mu + 1}.
\end{equation*}

\end{proof}
\end{theorem}

\begin{rmk} The proof of Theorem \ref{thm:PoA-smooth-games} can be easily extend to $\epsilon$-NEs and thus a performance lower bound can also be derived for $\epsilon$-NEs:
\begin{align*}
    -\epsilon \le (\mu+1) W(\pi^{\epsilon\textup{-}\NE})) - \lambda W(\pi^\star)\\
    \Longrightarrow W(\pi^{\epsilon\textup{-}\NE}))\ge \frac{\lambda}{\mu+1}W(\pi^\star) - \frac{\epsilon}{\mu+1}
\end{align*}
\end{rmk}

\subsection{A counter-example for general Markov games}

Theorem \ref{thm:PoA-smooth-games} suggests that as long as the stage rewards of the Markov game form a {$(\lambda, \mu)$-generalized smooth game} for every horizon $h$, then the PoA bound analysis can be generalized from the static normal form game to the Markov game. Given
this positive result, a natural follow-up question is, whether the statement can be generalized to general Markov games in the sense that stage-wise PoA lower bound can also bound the PoA for the Markov game which is not necessarily a smooth game. Unfortunately, this statement is not necessarily true. 

 \begin{wrapfigure}{r}{.22\textwidth}
\centering
\vspace{-15pt}
\includegraphics[width=.22\textwidth]{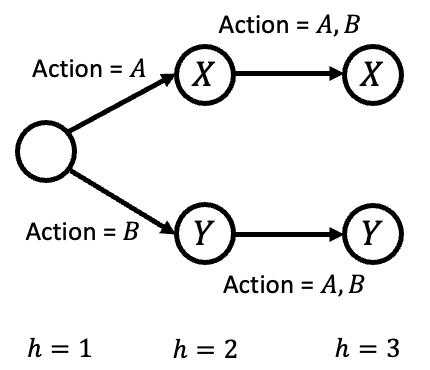}
\vspace{-15pt}
\caption{State Transition}
\vspace{-5pt}
\label{fig:state-transition}
\end{wrapfigure}We can construct a simple counter-example with $n=2, H=3$. In this counter-example, the reward functions $r_{i,h}$'s only depend on states (but not actions), and for any $1 \le h\le H$, $\{r_{i,h}(\cdot)\}_{i=1}^n$ forms a normal form game on the state space with $\PoA(\{r_{i,h}(\cdot)\}_{i=1}^n)\ge 1/2$, yet the PoA for the Markov game $\PoA(\cM)= 7/16 < 1/2$. The counter-example is constructed as follows: both agent $1$ and $2$ follow the (local) state transition dynamics as shown in Figure \ref{fig:state-transition}. The reward table for $h=2$ and $h=3$ is shown as in Table \ref{tab:reward-h-2}, \ref{tab:reward-h-3}. It can be verified that the PoA is $1/2$ for Table \ref{tab:reward-h-2} (the worst NE is a mixed NE where both agents play $X$ and $Y$ with probability $1/2$) and $1$ for Table \ref{tab:reward-h-3} (there's a unique NE which is also the global optimal).  From another perspective, the Markov game is equivalent to a normal form game on the action space at $h=1$, whose reward table is shown in Table \ref{tab:reward-equivalence}. Thus the PoA for the dynamic game is the same as Table \ref{tab:reward-equivalence}, which is $\PoA(\cM) = 7/16 < 1/2$ (the worst NE is a mixed NE where both agents play $A$ with probability $3/4$ and $Y$ with probability $1/4$). 



\begin{table}[htbp]
\parbox{.3\linewidth}{
\centering
\begin{tabular}{c|cc}
     & X &Y  \\
     \hline
     X & 2 & 0\\
     Y & 0 & 2\\
\end{tabular}
\caption{Reward \\ table for $h=2$}
\label{tab:reward-h-2}
}
\parbox{.3\linewidth}{
\centering
\begin{tabular}{c|cc}
     & X &Y  \\
     \hline
     X & 0 & 1\\
     Y & 1 & 2\\
\end{tabular}
\caption{Reward\\ table for $h=3$}
\label{tab:reward-h-3}
}
\parbox{.3\linewidth}{
\centering
\begin{tabular}{c|cc}
     & A &B  \\
     \hline
     A & 2 & 1\\
     B & 1 & 4\\
\end{tabular}
\caption{Equivalence to a normal form game}
\label{tab:reward-equivalence}
}
\end{table}

%% file: CDC2023/sections/4_alg_complexity.tex
\section{LEARNING THE NASH EQUILIBRIUM FOR MARKOV POTENTIAL GAMES}
The PoA lower bound proved by the previous section suggests that as long as an algorithm finds a NE for a {$(\lambda, \mu)$-generalized smooth Markov game}, it is at most $\lambda/(\mu+1)$ suboptimal in terms of the performance ratio. However, finding a NE is known to be intractable for general sum games \cite{chen2009settling,daskalakis2009complexity}, thus as a starting point, in this paper we mainly focus on an important subclass, namely the Markov potential game, and propose a multi-agent soft policy iteration (MA-SPI) algorithm that learns a NE efficiently using samples.
\begin{defi}\textit{(Markov Potential Game (MPG, \cite{Macua18,Gonzalez13,ZhangMPG,Leonardos}))}\label{defi:MPG}
A Markov game $\cM$ defined as Equation \eqref{eq:MG-def} is a Markov potential game if there exists stage potential functions $\{\phi_h: \cS\times\cA\to \mathbb{R}\}_{h=1}^H$ such that for all agents' rewards and state-action pairs at all horizon $1\le h \le H$, 
\begin{equation}\label{eq:MPG-phi}
\begin{split}
        r_{i,h}(s_{i,h}', s_{-i,h}, a_{i,h}', a_{-i,h}) &- r_{i,h}(s_{i,h}, s_{-i,h}, a_{i,h}, a_{-i,h})\\
        = \phi_{h}(s_{i,h}', s_{-i,h}, a_{i,h}', a_{-i,h}) &- \phi_{h}(s_{i,h}, s_{-i,h}, a_{i,h}, a_{-i,h}).
\end{split}
\end{equation}    
\end{defi}
From the definition of MPG, we can define the total potential function as a function on the policy space
\begin{equation}\label{eq:MPG-Phi}
 \Phi(\pi):=  \bE_{s_{i,1}\sim\rho_i}^\pi \sum_{h=0}^{H} \phi_{h}(s_h,a_h),
\end{equation}
and it is not hard to verify that for any $(\pi_i,\pi_{-i}),(\pi_i',\pi_{-i})\in \Pi^\local$ the total potential function satisfies:
\begin{equation}\label{eq:MPG-def}
   J_i(\pi_i,\pi_{-i}) - J_i(\pi_i',\pi_{-i}) =\Phi(\pi_i,\pi_{-i}) - \Phi(\pi_i',\pi_{-i}).
\end{equation}

The notion of the Markov potential game is first proposed for continuous dynamical systems in \cite{Macua18,Gonzalez13}. The idea is then generalized to the Markov game setting in \cite{ZhangMPG,Leonardos}, where Equation \eqref{eq:MPG-def} is directly used as the definition of MPG. However, Equation \eqref{eq:MPG-def} 
is hard to check for general Markov games. In contrast, thanks to the decoupled dynamics and local policy class considered in this paper, the condition can be further simplified as Definition \ref{defi:MPG}, whose condition is defined over stage rewards, much easier to verify.

For the learning algorithm, we also make the following assumption on the MPG.
\begin{assump}[Sufficient Exploration]\label{assump:sufficient-exploration}
For any policy $\pi$, there exists a uniform constant ~$c > 0 $~ such that ~$d_{i,h}^\pi(s_i) > c, ~\forall~i,h,s_i$. 
\end{assump}
{Assumption \ref{assump:sufficient-exploration} requires that any state at any horizon $h$ has a positive probability of being visited. Similar assumptions are standard for proving convergence for sample-based RL algorithms (e.g., \cite{agarwal2021theory, Qu20}). We would like to note that Assumption \ref{assump:sufficient-exploration} is admittedly stronger than assumptions in \cite{agarwal2021theory} (they only require a certain discounted state visitation distribution to be positive, yet we require the state visitation distribution at \emph{every} horizon $h$ to be positive). We leave it to future work to extend this assumption. } 

\subsection{Algorithm design: multi-agent soft policy iteration}
Our algorithm can be summarized as the following soft policy iteration update: 
\begin{equation}\label{eq:inexact-soft-PI}
\begin{split}
       \pi_{i,h}^{(t+1)}(a_i|s_i) &=(1-\eta_t) \pi_{i,h}^{(t)}(a_i|s_i)+\\ &\eta_t \mathbf{1}\{a_i = \argmax_{a_i'} \hQ_{i,h}^{(t)}(s_i,a_i')\}.
\end{split}
\end{equation}
Here $\hQ_{i,h}^{(t)}$ is an estimation of the averaged-Q function $\oQ_{i,h}^{(t)}$ defined in Equation \eqref{eq:averaged-Q}. {Note that when the stepsize $\eta_t = 1$, the algorithm is equivalent to each agent executing the policy iteration algorithm, i.e., choosing greedy actions that maximize the current estimation of its own averaged-Q functions. However, for multi-agent systems, the greedy update with $\eta_t=1$ might cause miscoordination and thus fail to converge, hence we introduce the soft policy iteration by setting the next iteration's policy as a convex combination of the current policy and the greedy policy.}

Given Equation \eqref{eq:inexact-soft-PI}, the key part of the algorithm lies in how to estimate the averaged-Q functions. To start with, we present the following property of the averaged Q-functions (the proof is deferred to the Appendix):
\begin{lemma}\label{lemma:averaged-Q-bellman-eq}
    The averaged Q-functions $\oQ_{i,h}^\pi$ in Equation \eqref{eq:averaged-Q} satisfies the following Bellman equation:
    \begin{equation}\label{eq:bellman-eq}
    \begin{split}
        &\oQ_{i,h}^\pi(s_i,a_i) =\orh(s_i,a_i) +\\ & ~~\sum_{s_i',a_i'}P_{i,h}(s_i'|s_i,a_i)\pi_{i,h+1}(a_i'|s_i')\oQ_{i,h+1}^\pi(s_i',a_i'),
    \end{split}
    \end{equation}
    where
    \begin{equation*}
    \begin{split}
        &\orh(s_i,a_i) := \hspace{-8pt}\mathop{\bE}_{\substack{j\neq i\\ s_j\sim d_{j,\!h}^{\pi_j}(\cdot)}} \mathop{\bE}_{\substack{j\neq i\\ a_j\sim {\pi_j}(\cdot|s_j)}} \hspace{-8pt}r_{i,h}(s_i,s_{-i}, a_i, a_{-i})
    \end{split}
    \end{equation*}
\end{lemma}
The key idea of our sample-based algorithm is to estimate $\oQh(s_i,a_i)$ and then perform soft policy iteration (Equation \eqref{eq:inexact-soft-PI}). Given Equation \eqref{eq:bellman-eq}, we start by estimating $\orh(s_i,a_i)$ and $\Ph(s_i'|s_i,a_i)$ and then calculate the estimated averaged-Q functions using Equation \eqref{eq:bellman-eq}.

In our algorithm, there are two different types of data collection processes, namely the pre-data collection and the on-policy data collection. In the pre-data collection step, the dataset $D_K$ is collected by setting each agent's policy  the uniform random policy. The on-policy data collection step is collected by running a specific policy, for each policy $\pi^{(t)}$, the dataset $D_J^{(t)}$ is collected by implementing $\pi^{(t)}$. 

The estimated state transition probability $\hPh(s_i'|s_i,a_i)$ uses the dataset $D_K$ from the pre-data collection which contains $T_K$ samples $D_K = \{\{s_{i,h}^{(k)}, a_{i,h}^{(k)}\}_{h=1}^H\}_{k=1}^{T_K}\in D_K$, i.e., 
\begin{equation}\label{eq:P-hat}
\hPh(s_i'|s_i,a_i):=\!
\begin{cases}
\!\frac{\sum_{k=1}^{T_K} \!\mathbf{1}\{s_{i,h+1}^{(k)} \!\!=\! s_i', s_{i,h}^{(k)}\!=\!s_i, a_{i,h}^{(k)} \!=\! a_i\}}{\sum_{k=1}^{T_K}\mathbf{1}\{s_{i,h}^{(k)}=s_i, a_{i,h}^{(k)} = a_i\}},\!\\
\textup{if} ~\sum_{k=1}^{T_K}\mathbf{1}\{s_{i,h}^{(k)}=s_i, a_{i,h}^{(k)} = a_i\} \ge 1;\\
\mathbf{1}\{s_i'=s_i\}, \\
\textup{if} ~\sum_{k=1}^{T_K}\mathbf{1}\{s_{i,h}^{(k)}=s_i, a_{i,h}^{(k)} = a_i\} =0.
\end{cases}
\end{equation}

The averaged reward $\orh(s_i,a_i)$ for each given policy $\pi$ is estimated using the on-policy dataset $D_J$ which contains $T_J$ samples $D_J^{(t)}= \{\{s_{i,h}^{(k)},a_{i,h}^{(k)}\}_{h=1}^H\}_{k=1}^{T_J}$, i.e.
\begin{equation}
\hrh(s_i,a_i):=
\begin{cases} \label{eq:r-hat}
\!\frac{\sum_{k=1}^{T_J} \!r_{i,h}(s_i,a_i,s_{-i,h}^{(k)},a_{-i,h}^{(k)})\mathbf{1}\{s_{i,h}^{(k)}\!=\!s_i\}}{\sum_{k=1}^{T_J}\mathbf{1}\{s_{i,h}^{(k)}=s_i\}}\!,\\
\qquad\quad\textup{if} ~\sum_{k=1}^{T_J}\mathbf{1}\{s_{i,h}^{(k)}=s_i\} \ge 1; \\
0, \qquad\textup{if} ~\sum_{k=1}^{T_J}\mathbf{1}\{s_{i,h}^{(k)}=s_i\} =0.
\end{cases}
\end{equation}

Then we can replace $\Ph$ and $\orh$ in Equation \eqref{eq:bellman-eq} with $\hPh$ and $\hrh$ and get the equation for the estimated averaged-Q functions which can be solved by backward dynamic programming. 
\begin{equation}\label{eq:Q-hat}
\begin{split}
    &\hQ_{i,h}^\pi(s_i,a_i) =\hrh(s_i,a_i) +\\
    & ~~\sum_{s_i',a_i'}\hPh(s_i'|s_i,a_i)\pi_{i,h+1}(a_i'|s_i')\hQ_{i,h+1}^\pi(s_i',a_i'), \\
    &\hQ_{i,H+1}^\pi(s_i,a_i)=0.
\end{split}
\end{equation}

Lastly, the policy update is given by Equation \eqref{eq:inexact-soft-PI}. A full description of the multi-agent soft policy iteration (MA-SPI) algorithm is in Algorithm \ref{alg:soft-PI}

\begin{algorithm}[htbp]
\caption{Multi-Agent Soft Policy Iteration (MA-SPI)} 
\label{alg:soft-PI}
\begin{algorithmic}[1]
\STATE \textbf{Pre-Data Collection:} Get dataset $D_K:= \{\{s_{i,h}^{(k)},a_{i,h}^{(k)}\}_{h=1}^H\}_{k=1}^{T_K}$ by sampling $s_{i,1}^{(k)}\sim\rho_i, a_{i,h}^{(k)}\sim \textup{Unif}(\cA_i), 1\le h \le H.$ 
\STATE \textbf{Estimation of $\Ph(s_i'|s_i,a_i)$'s:} Estimate the state transition probabilities $\Ph(s_i'|s_i,a_i)$'s using Equation \eqref{eq:P-hat} and dataset $D_K$.
\FOR{$t = 1,2,\dots, T_G+1$}
\FOR{Agent $i = 1,2,\dots, n$}
\STATE \textbf{On-policy Data Collection:} Get dataset $D_J^{(t)}:= \{\{s_{i,h}^{(k)},a_{i,h}^{(k)}\}_{h=1}^H\}_{k=1}^{T_J}$ by sampling $s_{i,1}^{(k)}\sim\rho_i, a_{i,h}^{(k)}\sim \pi_{i,h}^{(t)}(\cdot|s_{i,h}^{(k)}), 1\le h\le H$.
\STATE \textbf{Estimation of $\orh(s_i,a_i)$'s:} Calculate $\hrh(s_i,a_i)$'s by Equation \eqref{eq:r-hat}.
\STATE \textbf{Estimation of $\oQh(s_i,a_i)$'s:}  Caculate $\hQh(s_i,a_i)$'s by Equation \eqref{eq:Q-hat}
\STATE \textbf{Policy Update:} Update policy $\pi_i^{(t+1)}$ by Equation \eqref{eq:inexact-soft-PI}.
\ENDFOR
\ENDFOR
\end{algorithmic}
\end{algorithm}

\subsection{Sample Complexity}
\begin{theorem}\label{thm:sample-complexity}
    Under Assumption \ref{assump:sufficient-exploration}, for any $\epsilon\le 1$ and any $\delta < 1$, by running Algorithm \ref{alg:soft-PI} with
    \begin{small}
    \begin{align*}
        T_G &\!\ge\!\frac{6\placeholder\left((\Phi_{\!\max}\!-\!\Phi_{\!\min}\!+\!1)/2 + \ln(8nH^{\frac{3}{2}}c^{\!-1}\epsilon^{\!-1})\right)^2}{c^2\epsilon^2} \\
        T_J &\!\ge\!\frac{2048n^2H^4\ln(8\delta^{-1}nHT_G\sum_i|S_i||A_i|)}{c^4\epsilon^2}\\
        T_K &\!\ge\! \frac{2048n^{\!2}\!\max_i\!|S_i\!|^{2}\!\max_i\!|A_i\!|^{2}\!H^{\!6}\ln(8\delta^{\!-1\!}\!nHT_G\sum_i\!\!|S_i\!|^{2}|A_i\!|)}{c^4\epsilon^2}\!,
    \end{align*}
    \end{small}
    it can be guaranteed that the output policies $\pi^{(t)}$'s from Algorithm \ref{alg:soft-PI} satisfy:
    \begin{equation}\label{eq:thm-eq-1}
    \sum_{t=1}^{T_G} \frac{1}{\sqrt{t}}\NEgap(\pi^{(t)}) \le \frac{\epsilon\sqrt{T_G}}{2} + \frac{\epsilon\sum_{t=1}^{T_G}\frac{1}{\sqrt{t}}}{2},
    \end{equation}
    which implies that:
    \begin{equation}\label{eq:thm-eq-2}
        \min_{1\le t\le T_G}\NEgap(\pi^{(t)}) \le\epsilon.
    \end{equation}
\end{theorem}

The sample complexity of Algorithm \ref{alg:soft-PI} is given by $T_G T_J + T_K$. From Theorem \ref{thm:sample-complexity}, by setting
\begin{align*}
    T_G &\sim\widetilde{O}\left(\frac{n^2H^3}{c^2\epsilon^2}\right), T_J\sim\widetilde{O}\!\left(\frac{n^2H^4}{c^4\epsilon^2}\right),\\T_K&\sim \widetilde{O}\left(\frac{n^2\max_i\!|S_i|^2\max_i\!|A_i|^2H^4}{c^2\epsilon^2}\right),
\end{align*}
where $\widetilde O$ hides the $\log$ factors, running Algorithm \ref{alg:soft-PI} can find an approximate NE with NE gaps smaller than $\epsilon$. Thus the complexity for finding an $\epsilon$-NE is $\widetilde O\left(\frac{n^4H^7}{c^6\epsilon^4} + \frac{n^2\max_i\!|S_i|^2\max_i\!|A_i|^2H^4}{c^2\epsilon^2}\right)$.

%% file: CDC2023/sections/5_proof_sketch.tex
\section{Proof Sketches for Theorem \ref{thm:sample-complexity}}
This section provides a brief proof sketch of our sample complexity result. The proof of the theorem can be decomposed into the following three major steps.

\noindent \textbf{Step 1: iteration complexity of MA-SPI.}
The proof of Theorem \ref{thm:sample-complexity} is based on the following main lemma on the iteration complexity of MA-SPI.
\begin{lemma}\label{lemma:inexact-soft-PI}
Suppose
\begin{equation*}
\begin{split}
    &\max_{s_i,a_i}\left|(\hQ^{\pi^{(t)}}_{i,h}-\oQ^{\pi^{(t)}}_{i,h})(s_i,a_i)\right|\le \epsilon_Q,\\
    &~~ \forall~1\le i\le n, 1\le h\le H, 1\le t\le T_G
\end{split}
\end{equation*}
then running Equation \eqref{eq:inexact-soft-PI} with $\eta_t = \frac{1}{\sqrt{\placeholder t}}$ will guarantee that
\begin{equation*}
\begin{split}
     &\sum_{t=1}^{T_G} \frac{1}{\sqrt{t}}\NEgap(\pi^{(t)}) \le\\ &\frac{\sqrt{\placeholder}}{c}\left(\Phi_{\max} \!-\! \Phi_{\min} \!+\!1\!+\! \ln(T_G)\right) + \frac{2nH\epsilon_Q}{c}\sum_{t=1}^{T_G}\frac{1}{\sqrt{t}},
\end{split}
\end{equation*}
which implies that
\begin{equation*}
\begin{split}
    &\min_{1\le t\le T_G}\NEgap(\pi^{(t)}) \le\\
    &\underbrace{\frac{\sqrt{\placeholder} \left(\Phi_{\max} - \Phi_{\min} + 1+ \ln(T_G)\right) }{2c\sqrt{T_G} }}_{\textup{Part I}}+ \underbrace{\frac{2nH\epsilon_Q}{c}}_{\textup{Part II}}.
\end{split}
\end{equation*}
\end{lemma}
Lemma \ref{lemma:inexact-soft-PI} suggests that the minimum NE-gap depends on two terms. The first term (Part I) diminishes to zero at rate $\widetilde O\left(\!\frac{1}{\sqrt{T_G}}\!\right)$, whereas the second term (Part II) is a bias term that is caused by the estimation error $\epsilon_Q$ of the averaged Q-functions. Thus, in order to reach an $\epsilon$-NE, in the proof we set both Part I and II to be smaller than $\frac{\epsilon}{2}$, and this gives a lower bound for $T_G$ and an upper bound for $\epsilon_Q$. Lemma \ref{lemma:inexact-soft-PI} is proved by a sufficient ascent lemma whose detail is in the Appendix.

\noindent \textbf{Step 2: sensitivity of $\epsilon_Q$ to the estimation errors.}
The rest of the proof focuses on the analysis of $\epsilon_Q$. In this step, we first bound $\epsilon_Q$ by the estimation errors of $\orh(s_i,a_i)$ and $P_{i,h}(s_i'|s_i,a_i)$, which is stated as follows:
\begin{lemma}\label{lemma:Q-hat-bound}
    Suppose
    \begin{align*}
        \left|\hrh(s_i,a_i)-\orh(s_i,a_i)\right|&\le\epsilon_r, \forall s_i\in\cS_i, a_i\in\cA_i,\! h,\\
        \left|\hPh(s_i'|s_i,a_i)\!-\!\Ph(s_i'|s_i,a_i)\right|&\!\le\!\epsilon_P, \!\forall s_i',s_i\!\in\!\cS_i, a_i\!\in\!\cA_i,\! h, 
    \end{align*}
    then $\hQ_{i,h}^\pi$ calculated from Equation \eqref{eq:Q-hat} satisfies:
    \begin{equation*}
        \max_{s_i,a_i}\left|(\hQh\!\!-\!\oQh)(s_i,a_i)\right| \!\le\! \epsilon_r(\!H\!+\!1\!-\!h\!) + \epsilon_P H(\!H\!+\!1\!-\!h\!)|S_i|.
    \end{equation*}
\end{lemma}
Lemma \ref{lemma:Q-hat-bound} is proved by induction on the Bellman equation (Equation \ref{eq:bellman-eq}).

\noindent \textbf{Step 3: bound the estimation errors of $\orh(s_i,a_i)$ and $P_{i,h}(s_i'|s_i,a_i)$}
The last step is bounding the estimation errors of the averaged-reward $\orh(s_i,a_i)$ and transition probability $P_{i,h}(s_i'|s_i,a_i)$.
\begin{lemma}\label{lemma:bound-r-P-hat}
Under Assumption \ref{assump:sufficient-exploration}, fix $s_i's_i,a_i,h$, for $\epsilon\le 1$,
\begin{equation*}
    \textstyle \Pr\left(\left|\hrh(s_i,a_i)-\orh(s_i,a_i)\right|\ge\epsilon \right) \le 4\exp\left(-\frac{\epsilon^2 c^2 T_J}{32}\right).
\end{equation*}
\begin{equation*}
   \textstyle \Pr\!\left(\!\left|\hPh(s_i'|s_i,a_i)\!\!-\!\!\Ph(s_i'|s_i,a_i)\right|\!\ge\!\epsilon \!\right) \!\le\! 4\!\exp\left(\!-\frac{\epsilon^2 c^2 T_K}{32|A_i|^2}\!\right)\!\!.
\end{equation*}
\end{lemma}
Finally, by combining the above three steps and setting $\epsilon_Q, \epsilon_r, \epsilon_P$ to appropriate values, we can calculate lower bounds for $T_G, T_J, T_K$ in order to get an $\epsilon$-NE.

%% file: CDC2023/sections/6_example.tex
\section{APPLICATION TO THE DYNAMIC COVERING GAME}\label{sec:covering}
\begin{wrapfigure}{r}{.2\textwidth}
    \centering
    \includegraphics[width=.2\textwidth]{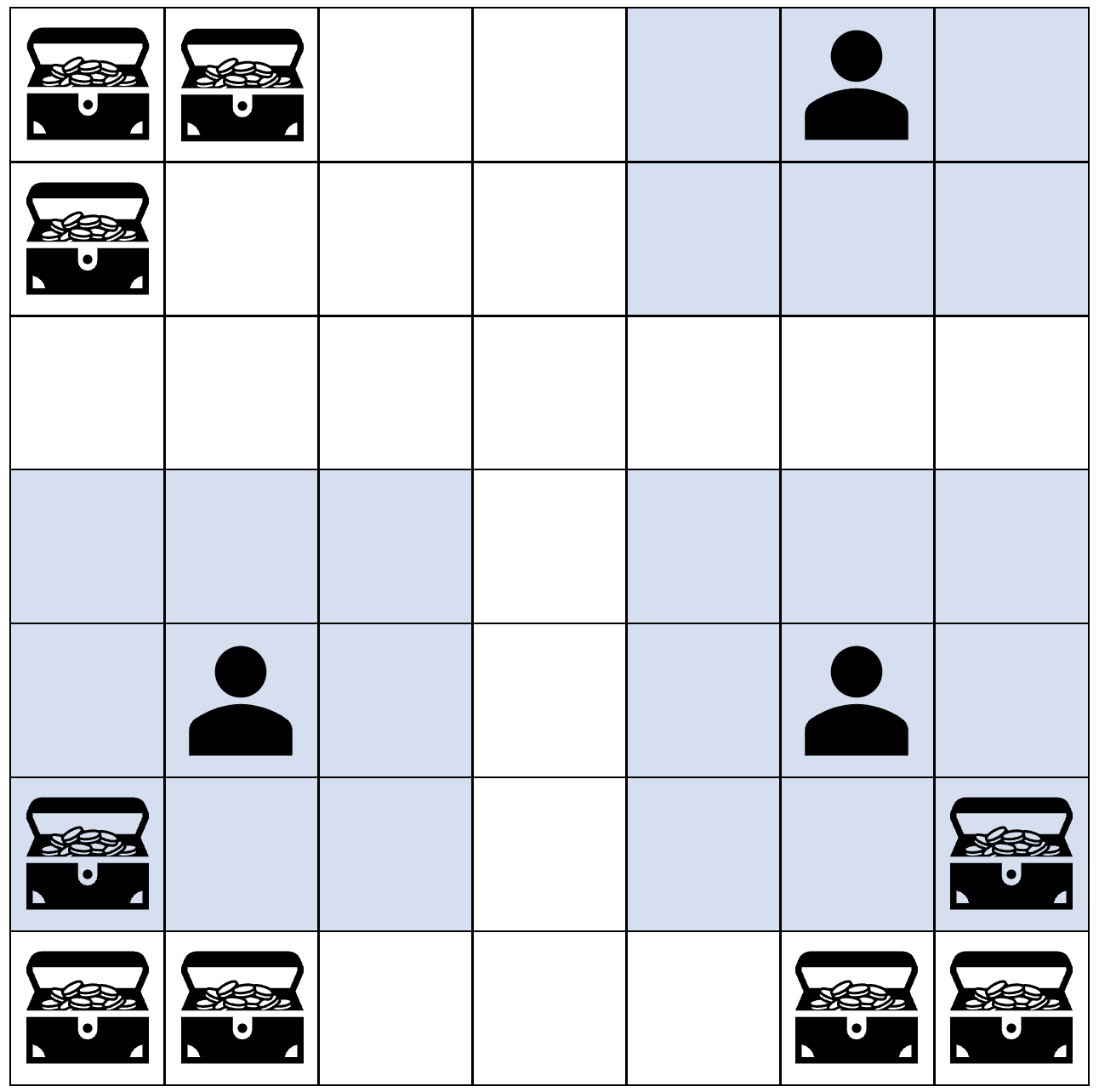}
    \caption{Multi-agent Dynamic Covering Game.}
    \label{fig:multi-agent dynamic covering game}
\end{wrapfigure}

 In this section, we use a simple dynamic covering game to illustrate our results. In a dynamic covering game, at each time, the agents' local states form a static covering game and the transition of the local states is governed by some local transition. While the PoA bounds for covering games in the static setting are well-established (e.g. \cite{roughgarden2015intrinsic,gairing2009covering}), not much is known for the dynamic setting. In this section, we look into three most commonly used reward functions (namely the identical-interest, marginal-contribution and utility-sharing), analyze their PoA bounds, and simulate the proposed soft policy iteration algorithm. 
 
 Specifically, we consider a group of $n$ collaborative agents in a covering game on a 2D grid of size $N$ ($[N]\times[N]$). For a grid $(j,k)$, there is a reward of value $w_{jk}$. The local state of each agent is its grid position, $s_{i,h}\in [N]\times [N]$. The local action $a_{i,h}$ represents the actions that an agent can take (such as moving up, down, left, right) on the grid. Each agent can have its own local station transition $P(s_{i,h+1}|s_{i,h},a_{i,h})$, either deterministic or stochstic. Given a position $s_{i,h}$ for agent $i$, the local area it can cover is denoted by $\hat{s}_{i,h}\subset [N]\times [N]$ which includes the grid $s_{i,h}$ and neighboring grids, e.g., the blue area in Figure~\ref{fig:multi-agent dynamic covering game}. \footnote{Without causing confusion, we use the notation $s_{i,h}$ to denote the local state of agent $i$ and $\hat{s}_{i,h}$ to denote the local area it can cover from the position.} The social welfare function $v$ at time step $h$ is how much treasure the agents cover at time step $h$:
\begin{equation*}
   \textstyle v(s_h):= \sum_{(j,k)\in \cup_{i=1}^n \hat{s}_{i,h}} w_{jk}
\end{equation*}
For different choices of stage reward function $r_{i,h}$, agents might converge to different NEs. Here we mainly look into the performance of the following three different types of reward functions:
\begin{enumerate}
    \item Identical-interest (II): $r_{i,h}(s_h) = v(s_h)$. The most natural idea might be setting all agents' stage rewards to be the same as the social welfare $v$, since this aligns with the ultimate objective. 
    \item Marginal-contribution (MC). $r_{i,h}(s_h) = v(s_h) - v(\emptyset,s_{-i,h})$. This type of reward design is also known as `marginal contribution' \cite{vetta2002nash} in static covering games/submodular games. The reward of agent $i$ is the change in the total reward if agent $i$ is removed. It is easy to see that this design is equivalent to setting $r_i(s_h) = \sum_{(j,k)\in \hat{s}_{i,h}}w_{jk}\mathbf{1}\{\textup{Only } i \textup{ covers } (j,k)\}$.
    \item Utility-sharing (US) \cite{gairing2009covering}. $r_i(s_h) = \sum_{(j,k)\in \hat{s}_{i,h}}w_{jk}\cdot\\f(\# \textup{agents that cover\ }(j,k))$, where $f$ is a monotonically decreasing function. Here the reward for covering $(j,k)$ is $w_{jk}f(\# \textup{agents that cover\ } (j,k))$. The agents will get a smaller reward for covering a grid that is already covered by many other agents. This type of reward assignment is also closely related to Shapley value and its variants. More discussions are available in \cite{shapley1951notes,DataShapley,BetaShapley}.
\end{enumerate}

\subsection{Smoothness and PoA Bounds}
In this section, we establish the smoothness and PoA bounds for all the above discussed reward functions.
\begin{lemma}\label{lemma:smoothness-covering-game}
\begin{enumerate}
    \item The sets of NEs for identical-interest and marginal-contribution are the same.
    \item For marginal-contribution, the Markov game is $(1,1)$-smooth, and thus the PoA for both marginal-contribution and identical-interest is bounded by $\PoA(\cM)\ge \frac{1}{2}$
    \item For untility-sharing, the Markov game is $(1,\mu_f)$-smooth, where $\mu_f = \sup_n nf(n) - f(n+1)$. More specifically, for 
    \vspace{-5pt}
    \begin{equation}\label{eq:utility-f}
       \textstyle f(n) = \frac{(n-1)!\sum_{i=n}^{+\infty} \frac{1}{i!}}{e-1},
    \end{equation}
    $\mu_f = \frac{1}{e-1}$ and thus the price of anarchy is bounded by $\PoA(\cM) \ge 1 - \frac{1}{e}$.
\end{enumerate}
\begin{proof}
    The first statement is trivial from the fact that II and MC share the same stage potential function, which can be chosen as the social welfare i.e., $\phi = v$. The smoothness condition for the second statement and third statement follows naturally from the smoothness analysis for the static covering game (see, e.g., \cite{roughgarden2015intrinsic} for the second statement and \cite{gairing2009covering} for the third statement). And the PoA bounds follow from Theorem \ref{thm:PoA-smooth-games}.
\end{proof}
\end{lemma}

\subsection{Numerical Simulations}
In this section, we test the MA-SPI algorithm (Algorithm \ref{alg:soft-PI}) in a dynamic covering game on a 2D grid of size 7. The state space of any agent $i$ consists of the $49$ grids. Treasure $w=1$ is set in grids $(0,0), (0,1), (1,0)$, $(0,5),(0,6),(1,6), (5,0),(6,0),(6,1)$, i.e. in three corners of the grid, as is shown in Figure~\ref{fig:multi-agent dynamic covering game}. 

There are $3$ agents in the game, whose initial positions are randomly selected uniformly among the $49$ grids. Each agent has $4$ actions, i.e. up, right, down and left. When any agent takes any action, the exact action is executed with probability $2/3$, a random action of the four is executed with probability $1/3$. When an action is executed, the agent transits to the corresponding grid or stays still if the corresponding grid is outside of the 2D plane. 

Every agent interacts with the environment for $H=10$ time steps. At every time step, every agent covers an area with size $3$ centered at its state, and receives reward according to the reward function and the treasures in its area of coverage. 

We first verify that a covering game forms a MPG with any of the three types of reward functions. For II and MC reward designs, it is obvious that the welfare function $v$ can serve as the stage potential function $\phi$ in Equation \eqref{eq:MPG-phi}. For utility-sharing, the stage potential function can be chosen as \cite{gairing2009covering}:
\begin{equation*}
    \textstyle \phi(s_h,a_h) = \sum_{(j,k)\in \hat{s}_h} w_{j,k} \sum_{l=1}^{\# \textup{agents that cover}\ (j,k)} f(l),
\end{equation*}
which satisfies the condition in Equation \eqref{eq:MPG-phi}. For any reward design, function $\phi$ gives the potential function $\Phi$ following Equation \eqref{eq:MPG-Phi}, and therefore proves that the game is an MPG.

Moreover, it can be verified by induction that this game satisfies Assumption \ref{assump:sufficient-exploration} with $c = \frac{1}{49}\times (\frac{1}{6})^{10}$. We simulate Algorithm~\ref{alg:soft-PI} in the environment described above. In the actual training process, the algorithm parameters are set to $T_G = 40, T_J = 800, T_K = 50000$, which is sufficient to show convergence for our algorithm.  For the US reward design, the utility function $f$ is set to Equation \eqref{eq:utility-f}.

\begin{figure}
    \centering
    \begin{minipage}{.49\columnwidth}
        \centering
        \includegraphics[width=\textwidth]{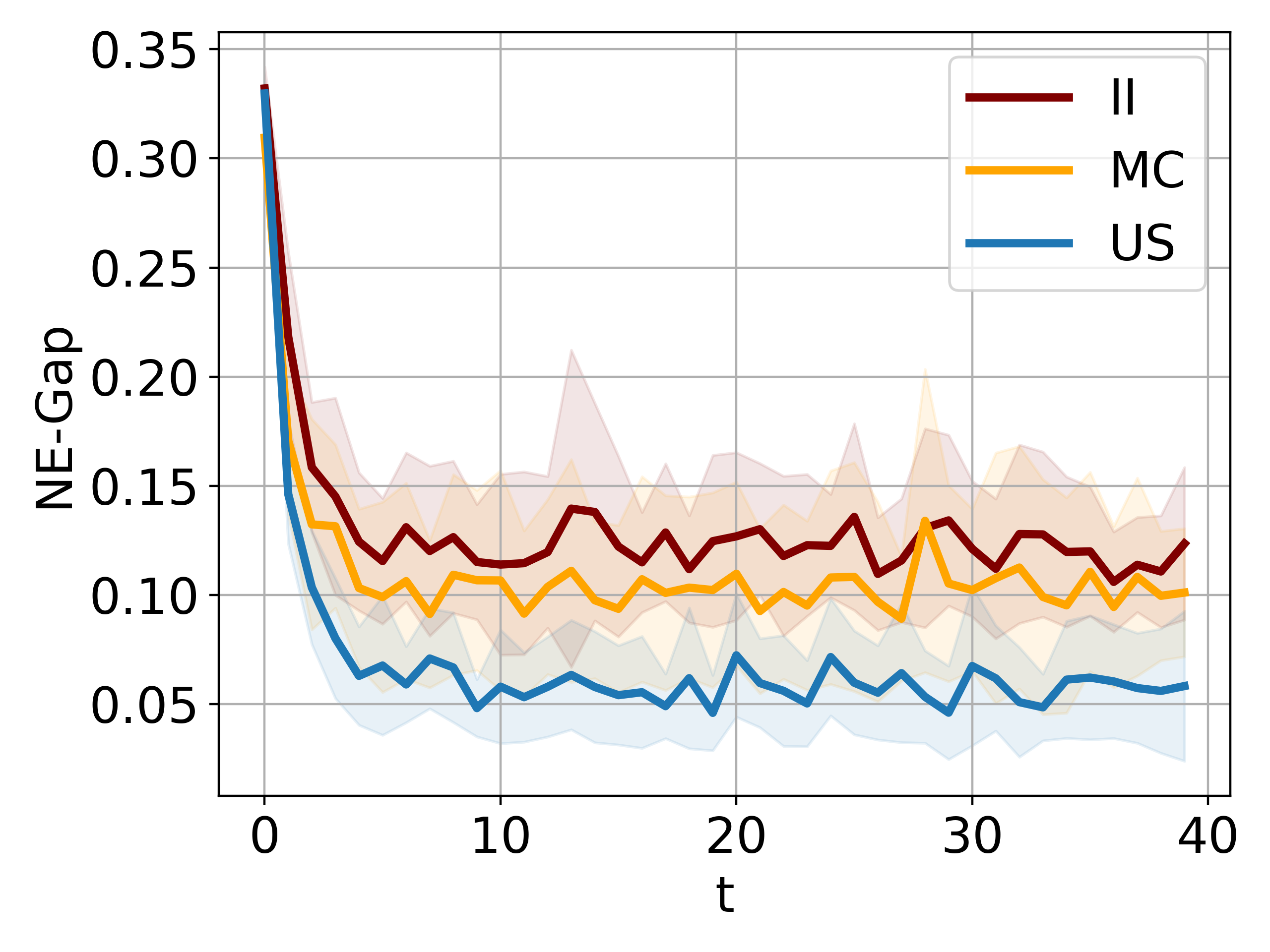}
    \end{minipage}%
    \begin{minipage}{.49\columnwidth}
        \centering
        \includegraphics[width=\textwidth]{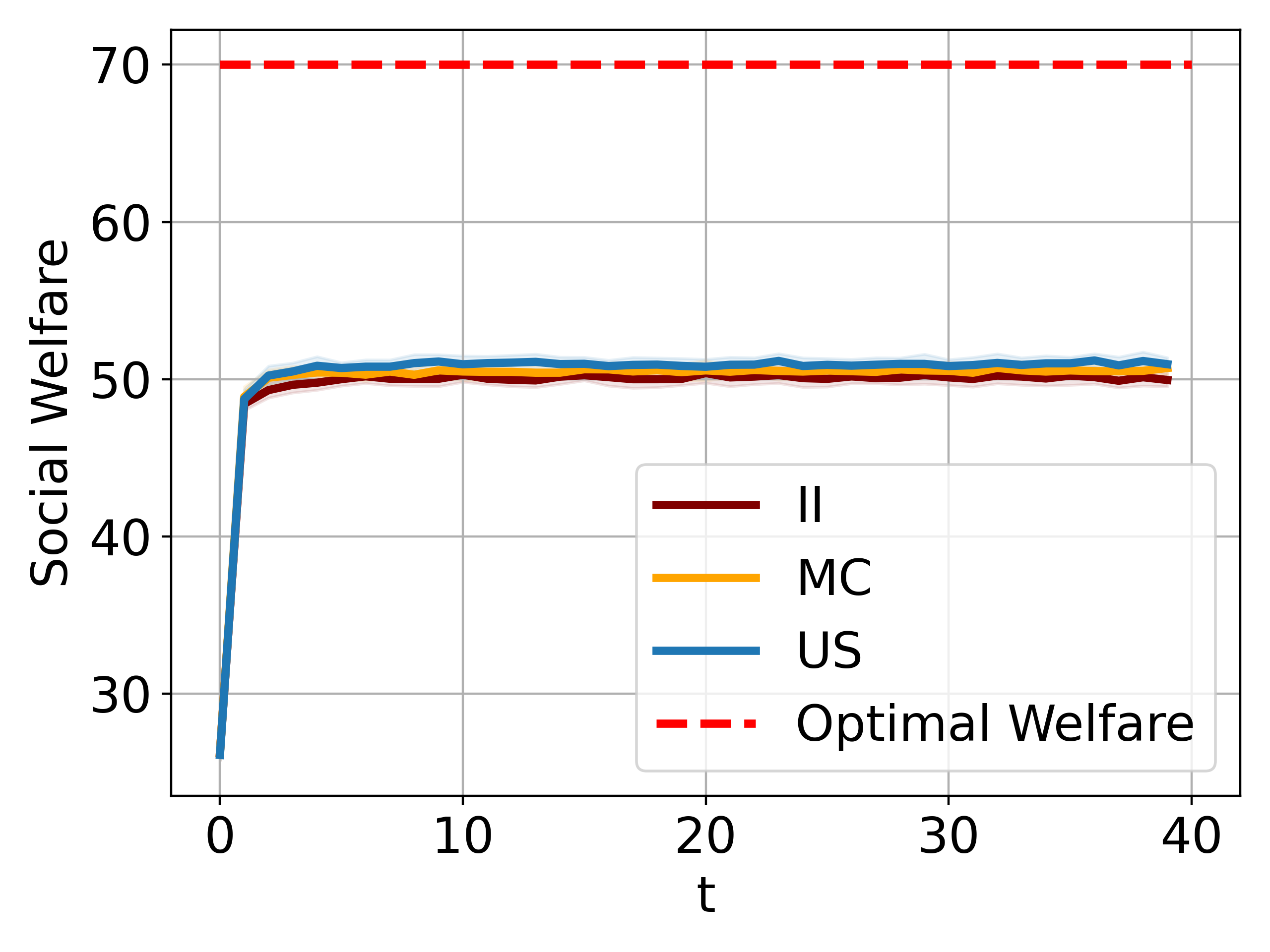}
    \end{minipage}
    \caption{\textbf{Left}: Social NE-Gap of Algorithm~\ref{alg:soft-PI}; \textbf{Right}: Social Welfare of Algorithm~\ref{alg:soft-PI};}
    \label{fig:sim}
    \vspace{-15pt}
\end{figure}

Figure~\ref{fig:sim} shows the NE-Gap and the social welfare in the training process for all three types of reward functions respectively. For all the reward designs, the NE-gap converges to almost zero, and the social welfare achieves around $50$, with the PoA being approximately $\text{PoA} \approx 50/70 > 1-\frac{1}{e}$, which aligns with Lemma~\ref{lemma:smoothness-covering-game}. 

%% file: CDC2023/sections/7_conclusion.tex
\section{CONCLUSIONS AND FUTURE WORKS}
This paper studies the price of anarchy (PoA) and efficient learning of a specific type of Markov games with decoupled dynamics. Firstly, the generalized smoothness condition is proposed to obtain a unified lower bound of PoA. For a specific type of Markov games known as the Markov potential games, we propose a distributed learning algorithm, multi-agent soft policy iteration (MA-SPI), which provably converges to a Nash equilibrium. Our results are also validated using a dynamic covering game example.

Our current results have several limitation as discussed in the paper, such as the restriction to local policy classes and the relatively strong assumption of sufficient exploration. The limitations open up several interesting avenues for future research. For example, it would be valuable to investigate how to achieve efficient learning and PoA analysis for broader equilibrium notions beyond Nash equilibrium, such as coarse correlated equilibria. In addition, we plan to explore algorithm design for more general policy classes where agents have access not only to their local states but also to their neighbors' information. Ultimately, we hope that this research will provide deeper insights into real-world multi-agent problems, such as multi-agent motion planning and sensor coverage problems. 

%

%% file: CDC2023/sections/8_appendix.tex
\appendix

\textbf{Notations:} Apart from the averaged Q-functions, we also define the averaged advantage functions as follows:
\begin{equation*}
    \oA^\pi_{i,h}(s_i,a_i):= \oQ^\pi_{i,h}(s_i,a_i) - \sum_{a_i'}\pi_{i,h}(a_i'|s_i)\oQ^\pi_{i,h}(s_i,a_i')
\end{equation*}
We also define the $\ell_1$-norm on the policy space as:
\begin{equation}\label{eq:policy-L-1}
    \|\pi_{i,h}' - \pi_{i,h}\|_1:= \max_{s_i} \|\pi_{i,h}'(\cdot|s_i)- \pi_{i,h}(\cdot|s_i)\|_1 = \max_{s_i}\sum_{a_i}\left|\pi_{i,h}'(a_i|s_i)- \pi_{i,h}(a_i|s_i)\right|
\end{equation}

\subsection{Performance difference lemma and policy gradient theorem}
\begin{lemma}\label{lemma:performance-difference}
\begin{equation*}
    J_i(\pi_i',\pi_{-i}) - J_i(\pi_i,\pi_{-i}) = \sum_{h=1}^H \sum_{s_i,a_i}d^{\pi_i'}_{i,h}(s_i)(\pi'_{i,h}(a_i|s_i) - \pi_{i,h}(a_i|s_i))\oQ^\pi_{i,h}(s_i,a_i)
\end{equation*}
\begin{proof}
Let $\pi':= (\pi_i',\pi_{-i})$. From the standard performance difference lemma we have that
\begin{align*}
    &\quad J_i(\pi_i',\pi_{-i}) - J_i(\pi_i,\pi_{-i}) = \sum_{h=1}^H \sum_{s,a}d_h^{\pi'}(s)(\pi_h'(a|s)-\pi_h(a|s))Q_{i,h}^\pi(s,a)\\
    &= \sum_{h=1}^H \sum_{s,a}d^{\pi_i'}_{i,h}(s_i)\prod_{j\neq i} d^{\pi_j}_{j,h}(s_j) (\pi'_{i,h}(a_i|s_i) - \pi_{i,h}(a_i|s_i))\prod_{j\neq i}\pi_{j,h}(a_j|s_j) Q_{i,h}^\pi(s_i,s_{-i}, a_i, a_{-i})\\
    &= \sum_{h=1}^H \sum_{s_i,a_i}d^{\pi_i'}_{i,h}(s_i)(\pi'_{i,h}(a_i|s_i) - \pi_{i,h}(a_i|s_i))\sum_{s_{-i},a_{-i}}\prod_{j\neq i} d^{\pi_j}_{j,h}(s_j)\prod_{j\neq i}\pi_{j,h}(a_j|s_j) Q_{i,h}^\pi(s_i,s_{-i}, a_i, a_{-i})\\
    &=  \sum_{h=1}^H \sum_{s_i,a_i}d^{\pi_i'}_{i,h}(s_i)(\pi'_{i,h}(a_i|s_i) - \pi_{i,h}(a_i|s_i))\oQ^\pi_{i,h}(s_i,a_i),
\end{align*}
which completes the proof.
\end{proof}
\end{lemma}

\subsection{Proof of Lemma \ref{lemma:averaged-Q-bellman-eq}}\label{sec:bellman-eq-proof}
\begin{proof}
    \begin{align*}
    &\quad \oQ^\pi_{i,h}(s_i,a_i) = \sum_{s_{-i}} \prod_{j\neq i} d_{j,h}^{\pi_j}(s_j)\sum_{a_{-i}}  \prod_{j\neq i}\pi_{j,h}(a_j|s_j) Q_{i,h}^\pi(s_i,s_{-i}, a_i, a_{-i})\\
    &= \sum_{s_{-i}} \prod_{j\neq i} d_{j,h}^{\pi_j}(s_j)\sum_{a_{-i}}  \prod_{j\neq i}\pi_{j,h}(a_j|s_j)\left(r_{i,h}(s,a)+\sum_{s',a'}P_h(s'|s,a)\pi_{h+1}(a'|s')Q_{i,h+1}^\pi(s',a')\right)\\
    &=\sum_{s_{-i}} \prod_{j\neq i} d_{j,h}^{\pi_j}(s_j)\sum_{a_{-i}}  \prod_{j\neq i}\pi_{j,h}(a_j|s_j)\left(r_{i,h}(s,a)+\sum_{s',a'}\prod_j P_{j,h}(s_j'|s_j,a_j)\pi_{j,h+1}(a_j'|s_j')Q_{i,h+1}^\pi(s',a')\right)\\
    &=\orh(s_i,a_i) + \sum_{s',a',s_{-i},a_{-i}}\prod_j\pi_{j,h+1}(a_j'|s_j')P_{j,h}(s_j'|s_j,a_j)\prod_{j\neq i} d_{j,h}^{\pi_j}(s_j)\pi_{j,h}(a_j|s_j)Q_{i,h+1}^\pi(s',a')\\
    &=\orh(s_i,a_i) + \sum_{s',a'}\prod_j\pi_{j,h+1}(a_j'|s_j') P_{i,h}(s_i'|s_i,a_i)\left(\sum_{s_{-i},a_{-i}}\prod_{j\neq i}P_{j,h}(s_j'|s_j,a_j)d_{j,h}^{\pi_j}(s_j)\pi_{j,h}(a_j|s_j)\right)Q_{i,h+1}^\pi(s',a')\\
    &= \orh(s_i,a_i) + \sum_{s',a'}\prod_j\pi_{j,h+1}(a_j'|s_j') P_{i,h}(s_i'|s_i,a_i)\prod_{j\neq i} d_{j,h+1}^{\pi_j}(s_j')Q_{i,h+1}^\pi(s',a')\\
    &=\orh(s_i,a_i) + \sum_{s_i',a_i'}P_{i,h}(s_i'|s_i,a_i)\pi_{i,h+1}(a_i'|s_i')\left(\sum_{s_{-i}',a_{-i'}}\prod_{j\neq i} d_{j,h+1}^{\pi_j}(s_j')\pi_{j,h+1}(a_j'|s_j')Q_{i,h+1}^\pi(s',a')\right)\\
    &= \orh(s_i,a_i) + \sum_{s_i',a_i'}P_{i,h}(s_i'|s_i,a_i)\pi_{i,h+1}(a_i'|s_i')\oQ_{i,h+1}^\pi(s_i',a_i')
    \end{align*}
\end{proof}

\subsection{Proof of Theorem \ref{thm:sample-complexity}}
\begin{proof}
    From Lemma \ref{lemma:inexact-soft-PI}, in order for Equation \eqref{eq:thm-eq-1} and Equation \eqref{eq:thm-eq-2} to hold, we can manually set
    \begin{align}
        \frac{\sqrt{\placeholder}(\Phi_{\max}-\Phi_{\min} + 1+\ln(T_G))}{2c\sqrt{T_G}}\le \frac{\epsilon}{2}\label{eq:TG-ineq}\\
        \frac{2nH\epsilon_Q}{c}\le \frac{\epsilon}{2}\label{eq:epsilonQ-ineq}.
    \end{align}
    From Equation \eqref{eq:TG-ineq} we get
    \begin{equation*}
        \frac{(\Phi_{\max}-\Phi_{\min} + 1+ \ln(T_G))}{\sqrt{T_G}} \le \frac{c\epsilon}{\sqrt{\placeholder}}.
    \end{equation*}
   From Lemma \ref{lemma:lnT/T-auxiliary}, we can set
    \begin{equation*}
        T_G \ge \frac{6\placeholder\left((\Phi_{\max}-\Phi_{\min} +1)/2 + \ln(4\sqrt{\placeholder}c^{-1}\epsilon^{-1})\right)^2}{c^2\epsilon^2},
    \end{equation*}
    which will guarantee Equation \eqref{eq:TG-ineq} to hold.

    From Equation \eqref{eq:epsilonQ-ineq}, we get $\epsilon_Q\le \frac{c\epsilon}{4nH}$. According to Lemma \ref{lemma:Q-hat-bound}, we can manually set
    \begin{equation*}
        \epsilon_r = \frac{c\epsilon}{8nH^2}\quad \epsilon_P=\frac{c\epsilon}{8nH^3\max_i|S_i|}.
    \end{equation*}
    From Lemma \ref{lemma:bound-r-P-hat}, by setting
    \begin{equation}\label{eq:T_J}
        T_J\ge \frac{32\ln(8\delta^{-1}nHT_G\sum_i|S_i||A_i|)}{c^2\epsilon_r^2} = \frac{2048n^2H^4\ln(8\delta^{-1}nHT_G\sum_i|S_i||A_i|)}{c^4\epsilon^2},
    \end{equation}
    it can be guaranteed that
    \begin{equation*}
    \Pr\left(\left|\widehat{r}_{i,h}^{\pi^{(t)}}(s_i,a_i)-\overline{r}_{i,h}^{\pi^{(t)}}(s_i,a_i)\right|\ge\epsilon_r \right) \le \frac{\delta}{2nHT_G\sum_i|S_i||A_i|}, ~~\forall i\in [n],h\in[H], s_i\in\cS_i, a_i\in\cA_i, 1\le t\le T_G,
    \end{equation*}
    thus
    \begin{equation*}
        \Pr\left(\left|\widehat{r}_{i,h}^{\pi^{(t)}}(s_i,a_i)-\overline{r}_{i,h}^{\pi^{(t)}}(s_i,a_i)\right|\ge\epsilon_r,~\forall i\in [n],h\in[H], s_i\in\cS_i, a_i\in\cA_i, 1\le t\le T_G\right) \le \frac{\delta}{2}
    \end{equation*}
    Similarly, from Lemma \ref{lemma:bound-r-P-hat}, by setting
\begin{equation}\label{eq:T_K}
    T_K \ge \frac{32\max_i|A_i|^2\ln(8\delta^{-1}nHT_G\sum_i|S_i|^2|A_i|)}{c^2\epsilon_P^2} = \frac{2048n^2\max_i|S_i|^2\max_i|A_i|^2H^6\ln(8\delta^{-1}nHT_G\sum_i|S_i|^2|A_i|)}{c^4\epsilon^2},
\end{equation}
it can be guaranteed that 
\begin{equation*}
    \Pr\left(\left|\widehat{P}_{i,h}^{\pi^{(t)}}(s_i'|s_i,a_i)-\overline{P}_{i,h}^{\pi^{(t)}}(s_i'|s_i,a_i)\right|\ge\epsilon_r,~\forall i\in [n],h\in[H], s_i',s_i\in\cS_i, a_i\in\cA_i, 1\le t\le T_G\right) \le \frac{\delta}{2}.
\end{equation*}

Thus, with the choice of $T_J$ and $T_K$ as in Equation \eqref{eq:T_J} and Equation \eqref{eq:T_K}, from Lemma \ref{lemma:Q-hat-bound}, with probability at least $1-\delta$, 
\begin{equation*}
    \max_{s_i,a_i}\left|(\hQ^{\pi^{(t)}}_{i,h}-\oQ^{\pi^{(t)}}_{i,h})(s_i,a_i)\right|\le \epsilon_Q,~~\forall~i\in[n], h\in[H], 1\le t\le T_G,
\end{equation*}
Thus, applying Lemma \ref{lemma:inexact-soft-PI}, we get
\begin{align*}
    \sum_{t=1}^{T_G} \frac{1}{\sqrt{t}}\sum_{i=1}^n\NEgap_i(\pi^{(t)}) &\le \frac{\epsilon\sqrt{T_G}}{2} + \frac{\epsilon\sum_{t=1}^{T_G}\frac{1}{\sqrt{t}}}{2},\\
    \min_{1\le t\le T_G}\sum_{i=1}^n\NEgap_i(\pi^{(t)}) &\le\epsilon,
\end{align*}
which completes the proof.
\end{proof}

\subsection{Proof of Lemma \ref{lemma:inexact-soft-PI}}
\begin{proof}[Proof of Lemma \ref{lemma:inexact-soft-PI}]
From Lemma \ref{lemma:sufficient-descent},
\begin{align*}
    \Phi(\pi^{(t+1)}) - \Phi(\pi^{(t)})&\ge c\eta_t\sum_{i=1}^n \sum_{h=1}^H \max_{s_i,a_i}\oA^{(t)}_{i,h}(s_i,a_i) -2\eta_t nH\epsilon_Q -\placeholder\eta_t^2\\
    &\stackrel{\textup{Lemma \ref{lemma:NE-gap-bound}}}{\ge} c\eta_t \sum_{i=1}^n \NEgap_i(\pi^{(t)}) -2\eta_t nH\epsilon_Q  - \placeholder\eta_t^2\\
    &= \frac{c}{\sqrt{\placeholder t}} \sum_{i=1}^n\NEgap_i(\pi^{(t)}) -\frac{2nH\epsilon_Q}{\sqrt{\placeholder t}} -\frac{1}{t}.
\end{align*}

Using telescoping we have that
\begin{align*}
    &\Phi_{\max} - \Phi_{\min} \ge \frac{c}{\sqrt{\placeholder}}\sum_{t=1}^{T_G} \frac{1}{\sqrt{t}}\sum_{i=1}^n\NEgap_i(\pi^{(t)}) -\frac{2nH\epsilon_Q}{\sqrt{\placeholder}}\sum_{t=1}^{T_G}\frac{1}{\sqrt{t}}- \sum_{t=1}^{T_G} \frac{1}{t}\\
    \Longrightarrow~~ &\sum_{t=1}^{T_G} \frac{1}{\sqrt{t}}\sum_{i=1}^n\NEgap_i(\pi^{(t)}) \le \frac{\sqrt{\placeholder}}{c}\left(\Phi_{\max} - \Phi_{\min} + 1+\ln(T_G)\right) + \frac{2nH\epsilon_Q}{c}\sum_{t=1}^{T_G}\frac{1}{\sqrt{t}}\\
    \Longrightarrow~~ &\min_{1\le t\le T_G}\sum_{i=1}^n\NEgap_i(\pi^{(t)}) \le \frac{\sqrt{\placeholder}\left(\Phi_{\max} - \Phi_{\min} + 1+\ln(T_G)\right) }{c\sum_{t=1}^{T_G}\frac{1}{\sqrt{t}}}+ \frac{2nH\epsilon_Q}{c} \\
    &\qquad\qquad\qquad\qquad\qquad\quad\le\frac{\sqrt{\placeholder} \left(\Phi_{\max} - \Phi_{\min} + 1+ \ln(T_G)\right) }{2c\sqrt{T_G} }+ \frac{2nH\epsilon_Q}{c}
\end{align*}
\end{proof}

\begin{lemma}\label{lemma:sufficient-descent}
Suppose
\begin{equation*}
    \max_{s_i,a_i}\left|(\hQh-\oQh)(s_i,a_i)\right|\le \epsilon_Q,~~\forall~i, h,
\end{equation*}
then running Equation Equation \eqref{eq:inexact-soft-PI} with will guarantee that
    \begin{equation*}
        \Phi(\pi^{(t+1)}) - \Phi(\pi^{(t)})\ge c\eta_t\sum_{i=1}^n \sum_{h=1}^H \max_{s_i,a_i}\oA^{(t)}_{i,h}(s_i,a_i) - 2 \eta_t nH\epsilon_Q - 4n^2H^3\eta_t^2
    \end{equation*}
\begin{proof}
Define the following auxiliary variables $\widetilde\pi^i:= (\pi_1^{(t)} ,\dots, \pi_{i-1}^{(t)}, \pi_i^{(t+1)}, \dots, \pi_n^{(t+1)})$, then
    \begin{align*}
        \Phi(\pi^{(t+1)}) - \Phi(\pi^{(t)}) &= \sum_{i=1}^n \Phi(\widetilde \pi^i) - \Phi(\widetilde \pi^{i+1})\\
        &= \sum_{i=1}^n J_i(\widetilde \pi^i) - J_i(\widetilde \pi^{i+1})\\
        &\stackrel{\textup{Lemma \ref{lemma:performance-difference}}}{=} \sum_{i=1}^n \sum_{h=1}^H \sum_{s_i,a_i}d^{\widetilde\pi^i}_{i,h}(s_i)(\pi^{(t+1)}_{i,h}(a_i|s_i) - \pi^{(t)}_{i,h}(a_i|s_i))\oQ^{\widetilde\pi^{i+1}}_{i,h}(s_i,a_i)\\
        &= \underbrace{\sum_{i=1}^n \sum_{h=1}^H \sum_{s_i,a_i}d^{\widetilde\pi^i}_{i,h}(s_i)(\pi^{(t+1)}_{i,h}(a_i|s_i) - \pi^{(t)}_{i,h}(a_i|s_i))\oQ^{(t)}_{i,h}(s_i,a_i)}_{\textup{Part I}}\\
        &+ \underbrace{\sum_{i=1}^n \sum_{h=1}^H \sum_{s_i,a_i}d^{\widetilde\pi^i}_{i,h}(s_i)(\pi^{(t+1)}_{i,h}(a_i|s_i) - \pi^{(t)}_{i,h}(a_i|s_i))(\oQ^{\widetilde\pi^{i+1}}_{i,h}(s_i,a_i)-\oQ^{(t)}_{i,h}(s_i,a_i))}_{\textup{Part II}}.
    \end{align*}
From Equation Equation \eqref{eq:inexact-soft-PI}, we have that
\begin{align*}
    \textup{Part I} &= \eta_t \sum_{i=1}^n \sum_{h=1}^H \sum_{s_i}d^{\widetilde\pi^i}_{i,h}(s_i)\sum_{a_i}\left(\mathbf{1}\{a_i = \argmax_{a_i'} \hQ_{i,h}^{(t)}(s_i,a_i)\}  - \pi_{i,h}^{(t)}(a_i|s_i)\right)\oQ^{(t)}_{i,h}(s_i,a_i) \\
    &=\eta_t \sum_{i=1}^n \sum_{h=1}^H \sum_{s_i}d^{\widetilde\pi^i}_{i,h}(s_i)\sum_{a_i}\left(\mathbf{1}\{a_i = \argmax_{a_i'} \hQ_{i,h}^{(t)}(s_i,a_i)\} - \mathbf{1}\{a_i = \argmax_{a_i'} \oQ_{i,h}^{(t)}(s_i,a_i)\}\right)\oQ^{(t)}_{i,h}(s_i,a_i) \\
    &\quad +\eta_t \sum_{i=1}^n \sum_{h=1}^H \sum_{s_i}d^{\widetilde\pi^i}_{i,h}(s_i)\sum_{a_i}\left(\mathbf{1}\{a_i = \argmax_{a_i'} \oQ_{i,h}^{(t)}(s_i,a_i)\}  - \pi_{i,h}^{(t)}(a_i|s_i)\right)\oQ^{(t)}_{i,h}(s_i,a_i) \\
    &= \eta_t \sum_{i=1}^n \sum_{h=1}^H \sum_{s_i}d^{\widetilde\pi^i}_{i,h}(s_i)\left(\sum_{a_i}\mathbf{1}\{a_i = \argmax_{a_i'} \hQ_{i,h}^{(t)}(s_i,a_i)\}\hQ^{(t)}_{i,h}(s_i,a_i) - \mathbf{1}\{a_i = \argmax_{a_i'} \oQ_{i,h}^{(t)}(s_i,a_i)\}\oQ^{(t)}_{i,h}(s_i,a_i)\right) \\
    &\quad - \eta_t \sum_{i=1}^n \sum_{h=1}^H \sum_{s_i}d^{\widetilde\pi^i}_{i,h}(s_i)\sum_{a_i}\mathbf{1}\{a_i = \argmax_{a_i'} \hQ_{i,h}^{(t)}(s_i,a_i)\}\left(\hQ^{(t)}_{i,h}(s_i,a_i) - \oQ^{(t)}_{i,h}(s_i,a_i)\right) \\
    &\quad +\eta_t \sum_{i=1}^n \sum_{h=1}^H \sum_{s_i}d^{\widetilde\pi^i}_{i,h}(s_i)\left(\max_{a_i}\oQ^{(t)}_{i,h}(s_i,a_i)  - \sum_{a_i}\pi_{i,h}^{(t)}(a_i|s_i)\oQ^{(t)}_{i,h}(s_i,a_i)\right) \\
    & \ge \eta_t \sum_{i=1}^n \sum_{h=1}^H \sum_{s_i}d^{\widetilde\pi^i}_{i,h}(s_i)\left(\max_{a_i}\hQ^{(t)}_{i,h}(s_i,a_i) - \max_{a_i}\oQ^{(t)}_{i,h}(s_i,a_i)\right)
    - \eta_t n H \epsilon_Q \\
    &\quad +\eta_t \sum_{i=1}^n \sum_{h=1}^H \sum_{s_i}d^{\widetilde\pi^i}_{i,h}(s_i) \max_{a_i}\oA^{(t)}_{i,h}(s_i,a_i)\\
    &\ge c\eta_t\sum_{i=1}^n \sum_{h=1}^H \max_{s_i,a_i}\oA^{(t)}_{i,h}(s_i,a_i) - 2\eta_t n H \epsilon_Q
\end{align*}
\begin{align*}
    |\textup{Part II}|&\le \max_{i,h,s_i,a_i}|\oQ^{\widetilde\pi^{i+1}}_{i,h}(s_i,a_i)-\oQ^{(t)}_{i,h}(s_i,a_i)| \sum_{i=1}^n \sum_{h=1}^H \sum_{s_i,a_i}d^{\widetilde\pi^i}_{i,h}(s_i)|\pi^{(t+1)}_{i,h}(a_i|s_i) - \pi^{(t)}_{i,h}(a_i|s_i)|\\
    &\stackrel{\textup{Lemma \ref{lemma:smoothness-averaged-Q}}}{\le } \left(H \sum_{h=1}^H \sum_{i=1}^n\|\pi_{i,h}^{(t+1)} - \pi_{i,h}^{(t)}\|_1\right)\left(\sum_{h=1}^H \sum_{i=1}^n\|\pi_{i,h}^{(t+1)} - \pi_{i,h}^{(t)}\|_1\right)\\
    &\le 4n^2H^3\eta_t^2\qquad \qquad
    \textup{($\|\pi_{i,h}^{(t+1)} - \pi_{i,h}^{(t)}\|_1\le 2\eta_t$)}
\end{align*}
Combining the bounds we get
\begin{align*}
    \Phi(\pi^{(t+1)}) - \Phi(\pi^{(t)}) &\ge \textup{Part I} - |\textup{Part II}|\ge c\eta_t\sum_{i=1}^n \sum_{h=1}^H \max_{s_i,a_i}\oA^{(t)}_{i,h}(s_i,a_i) - 2\eta_t n H \epsilon_Q -  4n^2H^3\eta_t^2
\end{align*}
\end{proof}
\end{lemma}

\begin{lemma}\label{lemma:NE-gap-bound}
    \begin{equation*}
        \NEgap_i(\pi) \le \sum_{h=1}^H\max_{s_i,a_i}\oA^\pi_{i,h}(s_i,a_i)
    \end{equation*}
\begin{proof}
From performance difference lemma:
\begin{align*}
    J_i(\pi_i',\pi_{-i}) - J_i(\pi_i,\pi_{-i}) &= \sum_{h=1}^H \sum_{s_i,a_i}d^{\pi_i'}_{i,h}(s_i)(\pi'_{i,h}(a_i|s_i) - \pi_{i,h}(a_i|s_i))\oQ^\pi_{i,h}(s_i,a_i)\\
    &= \sum_{h=1}^H \sum_{s_i,a_i}d^{\pi_i'}_{i,h}(s_i)\pi'_{i,h}(a_i|s_i)\oA^\pi_{i,h}(s_i,a_i)\\
    &\le \sum_{h=1}^H \max_{s_i,a_i}\oA^\pi_{i,h}(s_i,a_i)\\
    \Longrightarrow~~ \NEgap_i(\pi) &=  \max_{\pi_i'}J_i(\pi_i',\pi_{-i}) - J_i(\pi_i,\pi_{-i}) \le \sum_{h=1}^H \max_{s_i,a_i}\oA^\pi_{i,h}(s_i,a_i)
\end{align*}
\end{proof}
\end{lemma}

\subsection{Proof of Lemma \ref{lemma:Q-hat-bound}}

\begin{proof}
We prove by backward induction, the inequality holds naturally for $H+1$, assume that it holds for $h+1, h+2, \dots, H+1$, then
    \begin{align*}
    \oQ_{i,h}^\pi(s_i,a_i) &= \orh(s_i,a_i) + \sum_{s_i',a_i'}\Ph(s_i'|s_i,a_i)\pi_{i,h+1}(a_i'|s_i')\oQ_{i,h+1}^\pi(s_i',a_i')\\
     \hQ_{i,h}^\pi(s_i,a_i) &= \hrh(s_i,a_i) + \sum_{s_i',a_i'}\hPh(s_i'|s_i,a_i)\pi_{i,h+1}(a_i'|s_i')\hQ_{i,h+1}^\pi(s_i',a_i')\\
     \Longrightarrow~(\hQh-\oQh)(s_i,a_i)&=(\hrh-\orh)(s_i,a_i) + \sum_{s_i',a_i'}\Ph(s_i'|s_i,a_i)\pi_{i,h+1}(a_i'|s_i')(\hQ_{i,h+1}^\pi-\oQ_{i,h+1}^\pi)(s_i',a_i')\\
     &\quad + \sum_{s_i',a_i'} (\hPh-\Ph)(s_i'|s_i,a_i)\pi_{i,h+1}(a_i'|s_i')\hQ_{i,h+1}^\pi(s_i',a_i')\\
     \Longrightarrow~\left|(\hQh-\oQh)(s_i,a_i)\right|&\le \epsilon_r+ \max_{s_i',a_i'}\left|(\hQ_{i,h+1}^\pi-\oQ_{i,h+1}^\pi)(s_i',a_i')\right| + \epsilon_P H|S_i|\\
     &\le \epsilon_r + \epsilon_r (H-h) + \epsilon_P H(H-h)|S_i| + \epsilon_P H|S_i|\\
     & = \epsilon_r(H+1-h) + \epsilon_P H(H+1-h)|S_i|,
\end{align*}
which completes the proof.
\end{proof}

\subsection{Proof of Lemma \ref{lemma:bound-r-P-hat}}
\begin{proof}

[Proof for reward estimation $\orh$]
\begin{align*}
    &\left\{\left|\hrh(s_i,a_i)-\orh(s_i,a_i)\right|\ge\epsilon\right\}\\
    &\subseteq \left\{\sum_{k=1}^Tr_{i,h}(s_i,a_i,s_{-i,h}^{(k)},a_{-i,h}^{(k)})\mathbf{1}\{s_{i,h}^{(k)}=s_i\} - (\orh(s_i,a_i)+\epsilon)\mathbf{1}\{s_{i,h}^{(k)}=s_i\}\ge0\right\}\mathop{\cup}\left\{\sum_{k=1}^{T}\mathbf{1}\{s_{i,h}^{(k)}=s_i\}=0 \right\}
\end{align*}
Let:
\begin{align*}
    X_{k}&:=r_{i,h}(s_i,a_i,s_{-i,h}^{(k)},a_{-i,h}^{(k)})\mathbf{1}\{s_{i,h}^{(k)}=s_i\} - (\orh(s_i,a_i)+\epsilon)\mathbf{1}\{s_{i,h}^{(k)}=s_i\}\\
    Y_{k}&:= X_{k} - \bE[X_k]
\end{align*}
Because $\epsilon\le 1$, it is easy to verify that $|X_k|\le2$. We have that:
$$|Y_k| \le |X_k|+\bE[|X_k|]\le 4.$$
Further,
\begin{align*}
    \bE[X_k]&= \bE[r_{i,h}(s_i,a_i,s_{-i,h}^{(k)},a_{-i,h}^{(k)})\mathbf{1}\{s_{i,h}^{(k)}=s_i\} - (\orh(s_i,a_i)+\epsilon)\mathbf{1}\{s_{i,h}^{(k)}=s_i\}]\\
    &= -\epsilon\bE\mathbf{1}\{s_{i,h}^{(k)}=s_i\}\le -\epsilon c
\end{align*}
the first line to the second line of the equation is derived by the fact that $s_i^{(k)}, a_i^{t}$ and $s_{-i}^{t}, a_{-i}^{(k)}$ are independent and that:
\begin{align*}
    &\quad \bE[r_{i,h}(s_i,a_i,s_{-i,h}^{(k)},a_{-i,h}^{(k)})\mathbf{1}\{s_{i,h}^{(k)}=s_i\}] = \bE[r_{i,h}(s_i,a_i,s_{-i,h}^{(k)},a_{-i,h}^{(k)})]\bE[\mathbf{1}\{s_{i,h}^{(k)}=s_i\}]\\
    &= \sum_{s_{-i}} \prod_{j\neq i} d_{j,h}^{\pi_j}(s_j)\sum_{a_{-i}}  \prod_{j\neq i}\pi_{j,h}(a_j|s_j) r_{i,h}(s_i,s_{-i}, a_i, a_{-i})\bE[\mathbf{1}\{s_{i,h}^{(k)}=s_i\}]\\
    &= \bE[\orh(s_i,a_i)\mathbf{1}\{s_t=s,a_{i,t}=a_t\}]
\end{align*}
According to Hoeffding inequality:
\begin{align*}
    &\quad \Pr\left(\left\{\sum_{k=1}^Tr_{i,h}(s_i,a_i,s_{-i,h}^{(k)},a_{-i,h}^{(k)})\mathbf{1}\{s_{i,h}^{(k)}=s_i\} - (\orh(s_i,a_i)+\epsilon)\mathbf{1}\{s_{i,h}^{(k)}=s_i\}\ge0\right\}\right) \\
    &= \Pr\left(\sum_{k=1}^{T} X_k\ge0\right)\\
    &=\Pr\left(\sum_{k=1}^{T} Y_k\ge -\sum_{k=1}^{T}\bE[X_k]\right)\\
    &\le \Pr\left(\frac{1}{T}\sum_{k=1}^{T} Y_k\ge -\epsilon c\right)\\
    &\le \exp\left(-\frac{\epsilon^2 c^2 T}{32}\right)
\end{align*}
Further,
\begin{align*}
    &\quad \Pr\left(\left\{\sum_{k=1}^{T}\mathbf{1}\{s_{i,h}^{(k)}=s_i\}=0 \right\}\right) \\
    & = \Pr\left(\left\{\sum_{k=1}^{T}\mathbf{1}\{s_{i,h}^{(k)}=s_i\} - \bE\left[\mathbf{1}\{s_{i,h}^{(k)}=s_i\}\right]=-\sum_{k=1}^T\bE\left[\mathbf{1}\{s_{i,h}^{(k)}=s_i\}\right] \right\}\right)\\
    &\le \Pr\left(\left\{\frac{1}{T}\sum_{k=1}^{T}\mathbf{1}\{s_{i,h}^{(k)}=s_i\} - \bE\left[\mathbf{1}\{s_{i,h}^{(k)}=s_i\}\right]\le-c\right\}\right)\\
    &\le \exp\left(-\frac{c^2T}{8}\right)
\end{align*}
Thus
\begin{align*}
    &\quad \Pr\left(\hrh(s_i,a_i)-\orh(s_i,a_i)\ge\epsilon  \right)\\ &\le \Pr\left(\left\{\sum_{k=1}^Tr_{i,h}(s_i,a_i,s_{-i,h}^{(k)},a_{-i,h}^{(k)})\mathbf{1}\{s_{i,h}^{(k)}=s_i\} - (\orh(s_i,a_i)+\epsilon)\mathbf{1}\{s_{i,h}^{(k)}=s_i\}\ge0\right\}\right) \\
    &\qquad + \Pr\left(\left\{\sum_{k=1}^{T}\mathbf{1}\{s_{i,h}^{(k)}=s_i\}=0 \right\}\right)\\
    &\le 2\exp\left(-\frac{\epsilon^2 c^2 T}{32}\right)
\end{align*}
Similarly
\begin{align*}
    \Pr\left(\orh(s_i,a_i)-\hrh(s_i,a_i)\ge\epsilon  \right)\le 2\exp\left(-\frac{\epsilon^2 c^2 T}{32}\right)\\
    \Longrightarrow \Pr\left(\left|\hrh(s_i,a_i)-\orh(s_i,a_i)\right|\ge\epsilon  \right)\le4\exp\left(-\frac{\epsilon^2 c^2 T}{32}\right)
\end{align*}
which completes the proof.

[Proof for transition probability estimation $\Ph$]
The proof resembles the proof for reward estimation.
\begin{align*}
    &\quad\left\{\hPh(s_i'|s_i,a_i)-\Ph(s_i'|s_i,a_i)\ge\epsilon \right\}\\
    &\subseteq \left\{\sum_{k=1}^{T} \mathbf{1}\{s_{i,h+1}^{(k)} = s_i', s_{i,h}^{(k)}=s_i, a_{i,h}^{(k)} = a_i\}- (\Ph(s_i'|s_i,a_i)+\epsilon)\mathbf{1}\{s_{i,h}^{(k)}=s_i, a_{i,h}^{(k)} = a_i\}\ge0\right\}\\
    &\qquad\mathop{\cup}\left\{\sum_{k=1}^{T}\mathbf{1}\{s_{i,h}^{(k)}=s_i, a_{i,h}^{(k)} = a_i\}=0 \right\}
\end{align*}
Let:
\begin{align*}
    X_{k}&:=\mathbf{1}\{s_{i,h+1}^{(k)} = s_i', s_{i,h}^{(k)}=s_i, a_{i,h}^{(k)} = a_i\}- (\Ph(s_i'|s_i,a_i)+\epsilon)\mathbf{1}\{s_{i,h}^{(k)}=s_i, a_{i,h}^{(k)} = a_i\}\\
    Y_{k}&:= X_{k} - \bE[X_k]
\end{align*}
Because $\epsilon\le 1$, it is easy to verify that $|X_k|\le2$. We have that:
$$|Y_k| \le |X_k|+\bE[|X_k|]\le 4.$$
Further,
\begin{align*}
    \bE[X_k]&= \bE[\mathbf{1}\{s_{i,h+1}^{(k)} = s_i', s_{i,h}^{(k)}=s_i, a_{i,h}^{(k)} = a_i\}- (\Ph(s_i'|s_i,a_i)+\epsilon)\mathbf{1}\{s_{i,h}^{(k)}=s_i, a_{i,h}^{(k)} = a_i\}]\\
    &= -\epsilon\bE\mathbf{1}\{s_{i,h}^{(k)}=s_i, a_{i,h}^{(k)} = a_i\}\le -\frac{\epsilon c}{|A_i|}
\end{align*}
According to Hoeffding inequality:
\begin{align*}
    &\quad \Pr\left(\left\{\sum_{k=1}^{T} \mathbf{1}\{s_{i,h+1}^{(k)} = s_i', s_{i,h}^{(k)}=s_i, a_{i,h}^{(k)} = a_i\}- (\Ph(s_i'|s_i,a_i)+\epsilon)\mathbf{1}\{s_{i,h}^{(k)}=s_i, a_{i,h}^{(k)} = a_i\}\ge0\right\}\right) \\
    &= \Pr\left(\sum_{k=1}^{T} X_k\ge0\right)\\
    &=\Pr\left(\sum_{k=1}^{T} Y_k\ge -\sum_{k=1}^{T}\bE[X_k]\right)\\
    &\le \Pr\left(\frac{1}{T}\sum_{k=1}^{T} Y_k\ge -\frac{\epsilon c}{|A_i|}\right)\\
    &\le \exp\left(-\frac{\epsilon^2 c^2 T}{32|A_i|^2}\right)
\end{align*}
Further,
\begin{align*}
    &\quad \Pr\left(\left\{\sum_{k=1}^{T}\mathbf{1}\{s_{i,h}^{(k)}=s_i,a_{i,h}^{(k)} = a_i\}=0 \right\}\right) \\
    & = \Pr\left(\left\{\sum_{k=1}^{T}\mathbf{1}\{s_{i,h}^{(k)}=s_i,a_{i,h}^{(k)} = a_i\} - \bE\left[\mathbf{1}\{s_{i,h}^{(k)}=s_i,a_{i,h}^{(k)} = a_i\}\right]=-\sum_{k=1}^T\bE\left[\mathbf{1}\{s_{i,h}^{(k)}=s_i,a_{i,h}^{(k)} = a_i\}\right] \right\}\right)\\
    &\le \Pr\left(\left\{\frac{1}{T}\sum_{k=1}^{T}\mathbf{1}\{s_{i,h}^{(k)}=s_i,a_{i,h}^{(k)} = a_i\} - \bE\left[\mathbf{1}\{s_{i,h}^{(k)}=s_i,a_{i,h}^{(k)} = a_i\}\right]\le-\frac{c}{|A_i|}\right\}\right)\\
    &\le \exp\left(-\frac{c^2T}{8|A_i|^2}\right)
\end{align*}
Thus
\begin{align*}
    &\quad \Pr\left(\hPh(s_i'|s_i,a_i)-\Ph(s_i'|s_i,a_i)\ge\epsilon \right)\\ &\le\Pr\left(\left\{\sum_{k=1}^{T} \mathbf{1}\{s_{i,h+1}^{(k)} = s_i', s_{i,h}^{(k)}=s_i, a_{i,h}^{(k)} = a_i\}- (\Ph(s_i'|s_i,a_i)+\epsilon)\mathbf{1}\{s_{i,h}^{(k)}=s_i, a_{i,h}^{(k)} = a_i\}\ge0\right\}\right) \\
    &\qquad + \Pr\left(\left\{\sum_{k=1}^{T}\mathbf{1}\{s_{i,h}^{(k)}=s_i,a_{i,h}^{(k)} = a_i\}=0 \right\}\right)\\
    &\le 2\exp\left(-\frac{\epsilon^2 c^2 T}{32|A_i|^2}\right)
\end{align*}
Similarly
\begin{align*}
    \Pr\left(\Ph(s_i'|s_i,a_i)-\hPh(s_i'|s_i,a_i)\ge\epsilon \right)\le 2\exp\left(-\frac{\epsilon^2 c^2 T}{32|A_i|^2}\right)\\
    \Longrightarrow \Pr\left(\left|\Ph(s_i'|s_i,a_i)-\hPh(s_i'|s_i,a_i)\right|\ge\epsilon  \right)\le4\exp\left(-\frac{\epsilon^2 c^2 T}{32|A_i|^2}\right)
\end{align*}
which completes the proof.
\end{proof}

\subsection{Auxiliaries}
\begin{lemma}\label{lemma:smoothness-averaged-Q}
\begin{equation*}
    |\oQ_{i,h}^{\pi'}(s_i,a_i) - \oQ_{i,h}^\pi(s_i,a_i)|\le H \sum_{h'=1}^H \sum_{j=1}^n\|\pi_{j,h'}' - \pi_{j,h}\|_1
\end{equation*}
\begin{proof}
\begin{align*}
    &\quad \oQ_{i,h}^{\pi'}(s_i,a_i) - \oQ_{i,h}^\pi(s_i,a_i)\\
    &= \sum_{s_{-i}} \prod_{j\neq i} d_{j,h}^{\pi_j'}(s_j)\sum_{a_{-i}}  \prod_{j\neq i}\pi_{j,h}'(a_j|s_j) Q_{i,h}^{\pi'}(s_i,s_{-i}, a_i, a_{-i}) - \sum_{s_{-i}} \prod_{j\neq i} d_{j,h}^{\pi_j}(s_j)\sum_{a_{-i}}  \prod_{j\neq i}\pi_{j,h}(a_j|s_j) Q_{i,h}^\pi(s_i,s_{-i}, a_i, a_{-i})\\
    &= \underbrace{\sum_{s_{-i}} \left(\prod_{j\neq i} d_{j,h}^{\pi_j'}(s_j)- \prod_{j\neq i} d_{j,h}^{\pi_j}(s_j)\right)\sum_{a_{-i}}  \prod_{j\neq i}\pi_{j,h}'(a_j|s_j) Q_{i,h}^{\pi'}(s_i,s_{-i}, a_i, a_{-i})}_{\textup{Part I}}\\
    &\quad + \underbrace{\sum_{s_{-i}} \prod_{j\neq i} d_{j,h}^{\pi_j}(s_j)\sum_{a_{-i}}  \left(\prod_{j\neq i}\pi_{j,h}'(a_j|s_j)-\prod_{j\neq i}\pi_{j,h}(a_j|s_j)\right)Q_{i,h}^{\pi'}(s_i,s_{-i}, a_i, a_{-i})}_{\textup{Part II}}\\
    &\quad +\underbrace{\sum_{s_{-i}} \prod_{j\neq i} d_{j,h}^{\pi_j}(s_j)\sum_{a_{-i}}  \prod_{j\neq i}\pi_{j,h}(a_j|s_j)\left(Q_{i,h}^{\pi'}(s_i,s_{-i}, a_i, a_{-i})-Q_{i,h}^{\pi}(s_i,s_{-i}, a_i, a_{-i})\right)}_{\textup{Part III}}.
\end{align*}

\begin{align*}
    |\textup{Part I}|&\le H\sum_{s_{-i}} \left|\prod_{j\neq i} d_{j,h}^{\pi_j'}(s_j)- \prod_{j\neq i} d_{j,h}^{\pi_j}(s_j)\right|\sum_{a_{-i}}  \prod_{j\neq i}\pi_{j,h}'(a_j|s_j)\\
    &= H\sum_{s_{-i}} \left|\prod_{j\neq i} d_{j,h}^{\pi_j'}(s_j)- \prod_{j\neq i} d_{j,h}^{\pi_j}(s_j)\right|\\
    &\le H \sum_{j\neq i} \|d_{j,h}^{\pi'} - d_{j,h}^\pi\|_1\stackrel{\textup{Lemma \ref{lemma:smoothness-d}}}{\le} H\sum_{j\neq i
    }\sum_{\underline h=1}^{h-1}\|\pi_{j,\underline{h}}'-\pi_{j,\underline h}\|_1.
\end{align*}
\begin{align*}
    |\textup{Part II}|&\le \sum_{s_{-i}} \prod_{j\neq i} d_{j,h}^{\pi_j}(s_j)\sum_{a_{-i}}  \left|\prod_{j\neq i}\pi_{j,h}'(a_j|s_j)-\prod_{j\neq i}\pi_{j,h}(a_j|s_j)\right|H\\
    &\stackrel{\textup{Lemma \ref{lemma:L-1-norm-ineq}}}{\le} H\sum_{s_{-i}} \prod_{j\neq i} d_{j,h}^{\pi_j}(s_j)\sum_{j\neq i}\|\pi_{j,h}'(\cdot|s_j) - \pi_{j,h}(\cdot|s_j)\|_1\\
    &= H\sum_{j\neq i}\sum_{s_j}d_{j,h}^{\pi_j}(s_j)\|\pi_{j,h}'(\cdot|s_j) - \pi_{j,h}(\cdot|s_j)\|_1\\
    &\le H\sum_{j\neq i}\|\pi_{j,h}'-\pi_{j,h}\|_1
\end{align*}
\begin{align*}
    |\textup{Part III}|&\le \max_{s,a}|Q_{i,h}^{\pi'}(s,a) - Q_{i,h}^\pi(s,a)|\stackrel{\textup{Lemma \ref{lemma:smoothness-Q}}} {\le}H\sum_{\overline h=h+1}^H\sum_{j=1}^n\|\pi_{j,\overline h}' - \pi_{j,\overline h}\|_1.
\end{align*}
\end{proof}
Combining the bounds we get
\begin{align*}
     |\oQ_{i,h}^{\pi'}(s_i,a_i) - \oQ_{i,h}^\pi(s_i,a_i)| \le |\textup{Part I}| + |\textup{Part II}| + |\textup{Part III}|\le H \sum_{h'=1}^H \sum_{j=1}^n\|\pi_{j,h'}' - \pi_{j,h}\|_1.
\end{align*}
\end{lemma}

\begin{lemma}\label{lemma:smoothness-d}
    \begin{equation*}
        \|d_{i,h}^{\pi'} - d_{i,h}^\pi\|_1\le \sum_{\underline h=1}^{h-1}\|\pi_{i,\underline{h}}'-\pi_{i,\underline h}\|_1.
    \end{equation*}
    \begin{proof}
        We prove this by induction. Since $d_{i,1}^{\pi'} = d_{i,1}^{\pi}$, thus the inequality holds for $h=1$. Assume that the inequality holds for current $h$, then
        \begin{align*}
            d_{i,h+1}^{\pi}(s_i) &= \sum_{\overline s_i,a_i}P_h(s_i|\overline{s_i},a_i)\pi_{i,h}(a_i|\overline s_i)d_{i,h}^{\pi}(\overline s_i)\\
            d_{i,h+1}^{\pi'}(s_i) &= \sum_{\overline s_i,a_i}P_h(s_i|\overline{s_i},a_i)\pi_{i,h}'(a_i|\overline s_i)d_{i,h}^{\pi'}(\overline s_i)\\
            \Longrightarrow~~\sum_{s_i}|d_{i,h+1}^{\pi'}(s_i)-d_{i,h+1}^{\pi}(s_i)|&\le  \sum_{s_i,\overline s_i,a_i}P_h(s_i|\overline{s_i},a_i)\pi_{i,h}'(a_i|\overline s_i)|d_{i,h}^{\pi'}(\overline s_i)-d_{i,h}^{\pi}(\overline s_i)|\\
            &\quad + \sum_{s_i,\overline s_i,a_i}P_h(s_i|\overline{s_i},a_i)|\pi_{i,h}'(a_i|\overline s_i)-\pi_{i,h}(a_i|\overline s_i)|d_{i,h}^{\pi}(\overline s_i)\\
            &= \sum_{\overline s_i}|d_{i,h}^{\pi'}(\overline s_i)-d_{i,h}^{\pi}(\overline s_i)| + \sum_{\overline s_i,a_i}d_{i,h}^{\pi}(\overline s_i)|\pi_{i,h}'(a_i|\overline s_i)-\pi_{i,h}(a_i|\overline s_i)|\\
            &\le \|d_{i,h}^{\pi'} - d_{i,h}^\pi\|_1 + \|\pi_{i,h}' - \pi_{i,h}\|_1\\
            &\le  \sum_{\underline h=1}^{h}\|\pi_{i,\underline{h}}'-\pi_{i,\underline h}\|_1,
        \end{align*}
        which completes the proof.
    \end{proof}
\end{lemma}

\begin{lemma}\label{lemma:smoothness-Q}
    \begin{equation*}
        \max_{s,a}|Q_{i,h}^{\pi'}(s,a) - Q_{i,h}^\pi(s,a)|\le H\sum_{\overline h=h+1}^H\sum_{j=1}^n\|\pi_{j,\overline h}' - \pi_{j,\overline h}\|_1.
    \end{equation*}
    \begin{proof}
        Again, we prove this by induction. Since $Q_{i,H}^{\pi'} = Q_{i,H}^\pi$, the argument naturally holds for $h=H$. Assume that the argument holds for $h+1$, then
        \begin{align*}
            &\quad |Q_{i,h}^{\pi'}(s,a) - Q_{i,h}^\pi(s,a)|= \left|\sum_{\overline s, \overline a} P(\overline s|s,a)\pi_{h+1}'(\overline a|\overline s) Q_{i,h+1}^{\pi'}(\overline s, \overline a) - \sum_{\overline s, \overline a} P(\overline s|s,a)\pi_{h+1}(\overline a|\overline s) Q_{i,h+1}^{\pi}(\overline s, \overline a)\right|\\
            &\le\left|\sum_{\overline s, \overline a} P(\overline s|s,a)\left(\pi_{h+1}'(\overline a|\overline s) - \pi_{h+1}(\overline a|\overline s)\right)Q_{i,h+1}^{\pi'}(\overline s, \overline a) \right|+\left| \sum_{\overline s, \overline a} P(\overline s|s,a)\pi_{h+1}(\overline a|\overline s)\left(Q_{i,h+1}^{\pi'}(\overline s, \overline a)- Q_{i,h+1}^{\pi}(\overline s, \overline a)\right)\right|\\
            &\le H\sum_{\overline s}P(\overline s|s,a)\|\pi_{h+1}'(\cdot|\overline s) - \pi_{h+1}(\cdot|\overline s)\|_1 + \max_{\overline s, \overline a} |Q_{i,h+1}^{\pi'}(\overline s,\overline a) - Q_{i,h+1}^\pi(\overline s,\overline a)|\\
            &\le H\sum_{\overline s}P(\overline s|s,a)\sum_{j=1}^n\|\pi_{j,h+1}'(\cdot|\overline s)-\pi_{j,h+1}(\cdot|\overline s)\|_1 +  \max_{\overline s, \overline a} |Q_{i,h+1}^{\pi'}(\overline s,\overline a) - Q_{i,h+1}^\pi(\overline s,\overline a)|\\
            &\le H\sum_{j=1}^n\|\pi_{j,h+1}'-\pi_{j,h+1}\|_1 + \max_{\overline s, \overline a} |Q_{i,h+1}^{\pi'}(\overline s,\overline a) - Q_{i,h+1}^\pi(\overline s,\overline a)|\\
            &\le H\sum_{\overline h=h+1}^H\sum_{j=1}^n\|\pi_{j,\overline h}' - \pi_{j,\overline h}\|_1.
        \end{align*}
    \end{proof}
\end{lemma}

\begin{lemma}\label{lemma:lnT/T-auxiliary}
    For any $x>0$ and $\epsilon\le 1$, we have that any
    \begin{equation*}
        T \ge \frac{16\left(\frac{x+1}{2}+\ln(4\epsilon^{-1})\right)^2}{\epsilon^2}
    \end{equation*}
    satisfies the following inequality:
    \begin{equation*}
        \frac{x+1+\ln(T)}{\sqrt{T}}\le \epsilon
    \end{equation*}
\begin{proof}
    It is not hard to verify $\frac{x+1+\ln(T)}{\sqrt{T}}$ monotonically decreases with $T\geq 4$. Thus for $T \ge \frac{16\left(\frac{x+1}{2}+\ln(4\epsilon^{-1})\right)^2}{\epsilon^2}$,
    \begin{align*}
        \frac{x+1+\ln(T)}{\sqrt{T}}&\le \frac{\epsilon}{2x + 2+ 4\ln(4\epsilon^{-1})} \left(x + 1+\ln\left( \frac{16\left(\frac{x+1}{2}+\ln(4\epsilon^{-1})\right)^2}{\epsilon^2}\right)\right)\\
        &= \frac{\epsilon}{2x + 2+ 4\ln(4\epsilon^{-1})} \left(x + 1 + 2\ln\left(\frac{x+1}{2}+\ln(4\epsilon^{-1})\right)+2\ln(4\epsilon^{-1})\right)\\
        &\le \frac{\epsilon}{2x + 2+4\ln(4\epsilon^{-1})} \left(x + 1 + 2\left(\frac{x+1}{2}+\ln(4\epsilon^{-1})\right)+2\ln(4\epsilon^{-1})\right)\\
        &= \frac{\epsilon}{2x + 2+4\ln(4\epsilon^{-1})} \left(2x +2+4\ln(4\epsilon^{-1})\right)=\epsilon
    \end{align*}
\end{proof}
\end{lemma}

\begin{lemma}\label{lemma:L-1-norm-ineq}
For any two groups of probability measures $\{P_i\}_{i=1}^n, \{P_i'\}_{i=1}^n$ such that $P_i, P_i'$ are probability measures on a finite space $\Omega_i$ for $i=1,2,\dots,n$. Then
\begin{equation*}
\sum_{\{\omega_i\in\Omega_i\}_{i=1}^n} \left|\prod_{i=1}^n P_i'(\omega_i) - \prod_{i=1}^n P_i(\omega_i)\right| \le \sum_{i=1}^n \|P_i' - P_i\|_1,       
\end{equation*}
where the $\ell_1$-norm is defined as $\|P_i' - P_i\|_1:=\sum_{\omega_i\in\Omega_i}\left|P_i'(\omega) - P_i(\omega)\right|.$
\begin{proof}
    We'll prove by induction on $n$. When $n=1$, the inequality holds trivially. For $n\ge 2$, suppose that the inequality holds for $n-1$, then
    \begin{align*}
        &\quad\sum_{\{\omega_i\in\Omega_i\}_{i=1}^n} \left|\prod_{i=1}^n P_i'(\omega_i) - \prod_{i=1}^n P_i(\omega_i)\right| = \sum_{\omega_n\in\Omega_n} \sum_{\{\omega_i\in\Omega_i\}_{i=1}^{n-1}} \left|P_n'(\omega_n)\prod_{i=1}^{n-1} P_i'(\omega_i) - P_n(\omega_n)\prod_{i=1}^{n-1} P_i(\omega_i)\right|\\
        &\le \sum_{\omega_n\in\Omega_n}P_n'(\omega_n)\sum_{\{\omega_i\in\Omega_i\}_{i=1}^{n-1}} \left|\prod_{i=1}^{n-1} P_i'(\omega_i) - \prod_{i=1}^n P_i(\omega_i)\right| + \sum_{\omega_n\in\Omega_n} \left|P_n'(\omega_n) - P_n(\omega_n)\right|\sum_{\{\omega_i\in\Omega_i\}_{i=1}^{n-1}}\prod_{i=1}^n P_i(\omega_i)\\
       & = \sum_{\{\omega_i\in\Omega_i\}_{i=1}^{n-1}} \left|\prod_{i=1}^{n-1} P_i'(\omega_i) - \prod_{i=1}^{n-1} P_i(\omega_i)\right| + \sum_{\omega_n\in\Omega_n} \left|P_n'(\omega_n) - P_n(\omega_n)\right|\\
       &\le \sum_{i=1}^{n-1} \|P_i' - P_i\|_1 + \|P_n' - P_n\|_1 = \sum_{i=1}^n \|P_i' - P_i\|_1.
    \end{align*}
\end{proof}
\end{lemma}

%% file: main_CDC.bbl
\begin{thebibliography}{10}

\bibitem{roughgarden2015intrinsic}
T.~Roughgarden, ``Intrinsic robustness of the price of anarchy,'' {\em Journal
  of the ACM (JACM)}, vol.~62, no.~5, pp.~1--42, 2015.

\bibitem{Chandan19}
R.~Chandan, D.~Paccagnan, and J.~R. Marden, ``When smoothness is not enough:
  Toward exact quantification and optimization of the price-of-anarchy,'' in
  {\em 2019 IEEE 58th Conference on Decision and Control (CDC)},
  pp.~4041--4046, 2019.

\bibitem{isele2018navigating}
D.~Isele, R.~Rahimi, A.~Cosgun, K.~Subramanian, and K.~Fujimura, ``Navigating
  occluded intersections with autonomous vehicles using deep reinforcement
  learning,'' in {\em 2018 IEEE international conference on robotics and
  automation (ICRA)}, pp.~2034--2039, IEEE, 2018.

\bibitem{huttenrauch2019deep}
M.~H{\"u}ttenrauch, S.~Adrian, G.~Neumann, {\em et~al.}, ``Deep reinforcement
  learning for swarm systems,'' {\em Journal of Machine Learning Research},
  vol.~20, no.~54, pp.~1--31, 2019.

\bibitem{xia2021digital}
K.~Xia, C.~Sacco, M.~Kirkpatrick, C.~Saidy, L.~Nguyen, A.~Kircaliali, and
  R.~Harik, ``A digital twin to train deep reinforcement learning agent for
  smart manufacturing plants: Environment, interfaces and intelligence,'' {\em
  Journal of Manufacturing Systems}, vol.~58, pp.~210--230, 2021.

\bibitem{kofinas2018fuzzy}
P.~Kofinas, A.~Dounis, and G.~Vouros, ``Fuzzy q-learning for multi-agent
  decentralized energy management in microgrids,'' {\em Applied energy},
  vol.~219, pp.~53--67, 2018.

\bibitem{zhang2017social}
Y.~Zhang, B.~Song, and P.~Zhang, ``Social behavior study under pervasive social
  networking based on decentralized deep reinforcement learning,'' {\em Journal
  of Network and Computer Applications}, vol.~86, pp.~72--81, 2017.

\bibitem{resQ}
R.~Pina, V.~De~Silva, J.~Hook, and A.~Kondoz, ``Residual q-networks for value
  function factorizing in multi-agent reinforcement learning,'' 2022.

\bibitem{wang2021shaq}
J.~Wang, J.~Wang, Y.~Zhang, Y.~Gu, and T.-K. Kim, ``Shaq: Incorporating shapley
  value theory into multi-agent q-learning,'' {\em arXiv preprint
  arXiv:2105.15013}, 2021.

\bibitem{SongWhencanwe}
Z.~Song, S.~Mei, and Y.~Bai, ``When can we learn general-sum markov games with
  a large number of players sample-efficiently?,'' {\em CoRR},
  vol.~abs/2110.04184, 2021.

\bibitem{JinV-learning}
C.~Jin, Q.~Liu, Y.~Wang, and T.~Yu, ``V-learning - {A} simple, efficient,
  decentralized algorithm for multiagent {RL},'' {\em CoRR},
  vol.~abs/2110.14555, 2021.

\bibitem{ZhangMPG}
R.~Zhang, Z.~Ren, and N.~Li, ``Gradient play in multi-agent markov stochastic
  games: Stationary points and convergence,'' {\em CoRR}, vol.~abs/2106.00198,
  2021.

\bibitem{Leonardos}
S.~Leonardos, W.~Overman, I.~Panageas, and G.~Piliouras, ``Global convergence
  of multi-agent policy gradient in markov potential games,'' {\em CoRR},
  vol.~abs/2106.01969, 2021.

\bibitem{li2022congestion}
S.~H. Li, D.~Calderone, and B.~A{\c{c}}{\i}kme{\c{s}}e, ``Congestion-aware path
  coordination game with markov decision process dynamics,'' {\em IEEE Control
  Systems Letters}, vol.~7, pp.~431--436, 2022.

\bibitem{calderone2017infinite}
D.~Calderone and S.~Shankar, ``Infinite-horizon average-cost markov decision
  process routing games,'' in {\em 2017 IEEE 20th International Conference on
  Intelligent Transportation Systems (ITSC)}, pp.~1--6, IEEE, 2017.

\bibitem{koutsoupias1999worst}
E.~Koutsoupias and C.~Papadimitriou, ``Worst-case equilibria,'' in {\em Stacs},
  vol.~99, pp.~404--413, Springer, 1999.

\bibitem{youn2008price}
H.~Youn, M.~T. Gastner, and H.~Jeong, ``Price of anarchy in transportation
  networks: efficiency and optimality control,'' {\em Physical review letters},
  vol.~101, no.~12, p.~128701, 2008.

\bibitem{paccagnan2019utility}
D.~Paccagnan, R.~Chandan, and J.~R. Marden, ``Utility design for distributed
  resource allocation—part i: Characterizing and optimizing the exact price
  of anarchy,'' {\em IEEE Transactions on Automatic Control}, vol.~65, no.~11,
  pp.~4616--4631, 2019.

\bibitem{gairing2009covering}
M.~Gairing, ``Covering games: Approximation through non-cooperation,'' in {\em
  International Workshop on Internet and Network Economics}, pp.~184--195,
  Springer, 2009.

\bibitem{gkatzelis2016optimal}
V.~Gkatzelis, K.~Kollias, and T.~Roughgarden, ``Optimal cost-sharing in general
  resource selection games,'' {\em Operations Research}, vol.~64, no.~6,
  pp.~1230--1238, 2016.

\bibitem{hartline2014price}
J.~Hartline, D.~Hoy, and S.~Taggart, ``Price of anarchy for auction revenue,''
  in {\em Proceedings of the fifteenth ACM conference on Economics and
  computation}, pp.~693--710, 2014.

\bibitem{radanovic2019learning}
G.~Radanovic, R.~Devidze, D.~Parkes, and A.~Singla, ``Learning to collaborate
  in markov decision processes,'' in {\em International Conference on Machine
  Learning}, pp.~5261--5270, PMLR, 2019.

\bibitem{chen2022convergence}
D.~Chen, Q.~Zhang, and T.~T. Doan, ``Convergence and price of anarchy
  guarantees of the softmax policy gradient in markov potential games,'' {\em
  arXiv preprint arXiv:2206.07642}, 2022.

\bibitem{chen2009settling}
X.~Chen, X.~Deng, and S.-H. Teng, ``Settling the complexity of computing
  two-player nash equilibria,'' {\em Journal of the ACM (JACM)}, vol.~56,
  no.~3, pp.~1--57, 2009.

\bibitem{daskalakis2009complexity}
C.~Daskalakis, P.~W. Goldberg, and C.~H. Papadimitriou, ``The complexity of
  computing a nash equilibrium,'' {\em Communications of the ACM}, vol.~52,
  no.~2, pp.~89--97, 2009.

\bibitem{Macua18}
S.~V. Macua, J.~Zazo, and S.~Zazo, ``Learning parametric closed-loop policies
  for markov potential games,'' {\em CoRR}, vol.~abs/1802.00899, 2018.

\bibitem{Gonzalez13}
D.~Gonz{\'a}lez-S{\'a}nchez and O.~Hern{\'a}ndez-Lerma, {\em Discrete--time
  stochastic control and dynamic potential games: the Euler--Equation
  approach}.
\newblock Springer Science \& Business Media, 2013.

\bibitem{agarwal2021theory}
A.~Agarwal, S.~M. Kakade, J.~D. Lee, and G.~Mahajan, ``On the theory of policy
  gradient methods: Optimality, approximation, and distribution shift,'' {\em
  The Journal of Machine Learning Research}, vol.~22, no.~1, pp.~4431--4506,
  2021.

\bibitem{Qu20}
G.~Qu and A.~Wierman, ``Finite-time analysis of asynchronous stochastic
  approximation and $q$-learning,'' in {\em Proceedings of Thirty Third
  Conference on Learning Theory} (J.~Abernethy and S.~Agarwal, eds.), vol.~125
  of {\em Proceedings of Machine Learning Research}, pp.~3185--3205, PMLR,
  09--12 Jul 2020.

\bibitem{vetta2002nash}
A.~Vetta, ``Nash equilibria in competitive societies, with applications to
  facility location, traffic routing and auctions,'' in {\em The 43rd Annual
  IEEE Symposium on Foundations of Computer Science, 2002. Proceedings.},
  pp.~416--425, IEEE, 2002.

\bibitem{shapley1951notes}
L.~S. Shapley, ``Notes on the n-person game—ii: The value of an n-person
  game.(1951),'' {\em Lloyd S Shapley}, 1951.

\bibitem{DataShapley}
A.~Ghorbani and J.~Zou, ``Data shapley: Equitable valuation of data for machine
  learning,'' in {\em International Conference on Machine Learning},
  pp.~2242--2251, PMLR, 2019.

\bibitem{BetaShapley}
Y.~Kwon and J.~Zou, ``Beta shapley: a unified and noise-reduced data valuation
  framework for machine learning,'' {\em arXiv preprint arXiv:2110.14049},
  2021.

\end{thebibliography}
